\documentclass[journal,onecolumn]{IEEEtran}
\usepackage{cite}
\usepackage{amsmath,amssymb,amsthm}
\usepackage{mathtools}
\usepackage{mathrsfs}
\usepackage{amsfonts}
\usepackage{graphicx}
\usepackage{url}
\usepackage{color}
\usepackage{setspace}
\usepackage{makecell}
\usepackage{array}
\newcolumntype{L}[1]{>{\raggedright\arraybackslash}p{#1}}
\usepackage[utf8]{inputenc}
\usepackage{booktabs}
\usepackage{algorithm,algorithmic}
\usepackage{tikz}
\usetikzlibrary{arrows.meta,positioning}
\usepackage{hyperref}

\theoremstyle{plain}
\newtheorem{theorem}{Theorem}[section]
\newtheorem{lemma}[theorem]{Lemma}
\newtheorem{corollary}[theorem]{Corollary}
\newtheorem{proposition}[theorem]{Proposition}
\newtheorem{conjecture}[theorem]{Conjecture}
\newtheorem{comptheorem}[theorem]{Computational Theorem}
\theoremstyle{definition}
\newtheorem{definition}[theorem]{Definition}
\theoremstyle{remark}
\newtheorem{example}[theorem]{Example}
\newtheorem{remark}[theorem]{Remark}

\newcommand{\F}{\mathbb{F}}
\newcommand{\Ftwo}{\mathbb{F}_2}
\newcommand{\Ffour}{\mathbb{F}_4}
\DeclareMathOperator{\tr}{tr}
\DeclareMathOperator{\wt}{wt}
\DeclareMathOperator{\swt}{swt}
\DeclareMathOperator{\rank}{rank}
\DeclareMathOperator{\supp}{supp}
\DeclareMathOperator{\rev}{rev}
\DeclareMathOperator{\Sp}{Sp}

\title{Ancestor Lifting of Quasi-Cyclic Additive Codes and Certified Quantum Codes}
\author{Liangdong~Lu, Ran~Xu, Ruipan~Yang, Qiang~Fu%
\thanks{L. Lu is with the Department of Basic Science, Air Force Engineering University, Xi'an, Shaanxi, P. R. China (e-mail: kelinglv@163.com). R. Xu is with the Department of Basic Science, Air Force Engineering University, Xi'an, Shaanxi, P. R. China (e-mail: 2551193818@qq.com). R. Yang is with the Department of Basic Science, Air Force Engineering University, Xi'an, Shaanxi, P. R. China (e-mail: yangruipan@aliyun.com).}}
\date{}

\begin{document}

\pagestyle{plain}
\markboth{}{}

\maketitle

\begin{abstract}
Additive codes over $\Ffour$ that are self-orthogonal under the trace inner product correspond to binary quantum stabilizer codes. We study 1-generator $(1,2)$-quasi-cyclic (QC) additive codes and their exact quantum distances. For composite lengths $n=\ell m$, an ancestor-lifting theorem transports dual words of the folded ancestor to the length-$n$ code. A two-layer decomposition for $\ell=2$ and a CRT splitting for odd $\ell$ yield dimension laws and explain equality, escape, and mixed-branch behavior in the certified examples. A palindromic criterion ensures symplectic self-orthogonality, while reversal pairing and a symplectic Gram rank formula describe broader construction families. Exact distances are established by information-set decoding screening, Brouwer--Zimmermann computation on the additive dual, and independent witness verification. We certify twenty-one record-breaking binary quantum codes across twenty entries of Grassl's code table, with twenty-eight certified members in total. Propagation improves thirty-four further entries, and explicit constructions are supplied for three table-tying GV entries marked Code Missing Construction. Six additional candidates are reported only with rigorous upper bounds and heuristic screening evidence. The results connect divisor-lattice structure in additive quasi-cyclic codes with reproducible certification of exact quantum distances.
\end{abstract}



\section{Introduction}

\IEEEPARstart{Q}{uantum} stabilizer codes $[[n,k,d]]_2$ are additive codes over $\Ffour$ that are self-orthogonal under the trace inner product~\cite{crss}. Improving the minimum distances in Grassl's code table~\cite{grassl} is a longstanding goal of quantum coding theory. Two difficulties dominate this goal. At small lengths, the search space has been thoroughly explored. At large lengths, exact minimum-distance computation becomes expensive. Certification is equally important: explicit logical operators give upper bounds, whereas lower bounds require exhaustive or exact computation. Uncertified candidates can fail under scrutiny, and we analyze two archived examples (Example~\ref{ex:60} in Section~\ref{sec:palin} and Remark~\ref{rem:72candidate}).

The systematic construction of quantum codes from classical cyclic codes began with the quantum BCH codes of Grassl and Beth~\cite{grasslbeth,grasslbeth00} and the quaternary cyclic and constacyclic constructions of Lin~\cite{lin04}. Li and Li gave elementary binary constructions of quantum codes of minimum distance three and four~\cite{lili04tit}. They also showed that binary cyclic codes, including repeated-root codes of even lengths, yield quantum codes through the Calderbank--Shor--Steane (CSS) and Steane constructions~\cite{lili04ijqi}. Aly, Klappenecker, and Sarvepalli determined when a BCH code contains its dual and derived quantum BCH parameters directly from the design parameters~\cite{aly}. These ideas grew into a broad literature on quantum codes from constacyclic and negacyclic codes, represented by the quantum negacyclic families of Kai, Zhu, and Tang~\cite{kai13} and the nonbinary quantum BCH constructions of La Guardia~\cite{laguardia14}, among many others.

A second line of research uses the stabilizer correspondence itself. Additive codes over $\Ffour$ and binary stabilizer codes are two views of the same objects~\cite{crss}, so advances in additive quaternary codes transfer directly to quantum codes. Gulliver and Kim constructed extremal additive self-dual codes over $\Ffour$ from circulant matrices~\cite{gulliverkim04}. Huffman developed the algebraic theory of additive cyclic codes over $\Ffour$~\cite{huffman}, and Bierbrauer established the general theory of additive cyclic codes together with their quantum stabilizer specializations~\cite{bierbrauer07,bierbrauer12}. Shi, Wu, and Sol\'e proved that asymptotically good additive cyclic codes exist~\cite{shiwusole18}. Dastbasteh and Shivji gave a polynomial representation of additive cyclic codes whose nearly self-orthogonal members over $\Ffour$ yield record-breaking binary quantum codes~\cite{dastbasteh25}.

Quasi-cyclic codes are especially fertile for this correspondence. Ezerman, Ling, \"Ozkaya, and Sol\'e applied quantum Construction~X to quasi-cyclic codes over $\Ffour$ and $\F_9$~\cite{ezerman-qc}. Ezerman, Grassl, Ling, \"Ozbudak, and \"Ozkaya characterized nearly self-orthogonal quasi-twisted codes and reported numerous record-breaking quantum codes~\cite{ezerman-nso}. Guan and his collaborators constructed binary quantum codes from quasi-cyclic codes~\cite{guan} and established necessary and sufficient symplectic self-orthogonality conditions for quasi-cyclic codes of even index~\cite{guan-symp}. Related constructions include nearly self-orthogonal codes~\cite{lisonek14,degwekar15}, Construction~X4~\cite{liux4}, multi-twisted codes~\cite{aydin23}, and extended quasi-cyclic codes~\cite{lv21}. On the algebraic side, Yu and Zhu construct quantum codes from group rings~\cite{yuzhu24}, and Benjwal and Bhaintwal analyze duals of 1-generator quasi-cyclic codes~\cite{benjwal24}. Most of these works consider linear codes or larger fields, and exact distance certification is rarely the main objective. Our additive 1-generator quasi-cyclic family connects additive theory with quasi-cyclic construction practice, and our pipeline certifies exact distances.

We report the following contributions, theory first and codes as its application.

1) \emph{Ancestor lifting and the divisor-lattice distance law} (Section~\ref{sec:lift}). For $n = \ell m$, the lift map $L: R_m \to R_n$ transports dual words of the folded ancestor code $\langle(\varphi a, \varphi b)\rangle$ to the length-$n$ code, with symplectic weight multiplied by $\ell$; it yields a logical operator whenever the lifted word lies outside $\mathcal C$ (Theorem~\ref{thm:lift}). This lift repeats ancestor blocks along the divisor lattice of $n$ and is distinct from the homological lifts of CSS codes~\cite{guemard}. As a corollary, divisibility $(x^m-1) \mid b$ forces the dual bound $d(\mathcal C^{\perp_{ts}}) \le n/m$, and yields $d_Q \le n/m$ when the corresponding lift lies outside $\mathcal C$ (Corollary~\ref{cor:div}). For $\ell=2$, a two-layer decomposition isolates an obstruction map $\widehat{R}$ coupling the layers (Theorem~\ref{thm:twolayer}) and yields an exact dimension law and distance formula (Corollary~\ref{cor:dimlaw}); for odd $\ell$, the layers split by CRT without a compatibility condition (Theorem~\ref{thm:crt}). The ancestor relation is verified exactly for the certified codes with a nontrivial fold listed in Table~\ref{tab:ancestor}: nine satisfy $d_Q(n) = (n/m)\, d_Q(m)$ at every divisor level, seven have minimum distance exactly one below the corresponding ancestor-lifting bound, and $[[63,7,14]]_2$ is at least seven below every divisor-lattice lift bound. The certified $[[91,15,17]]_2$ and $[[93,15,17]]_2$ codes attain their distances in the mixed branch below dual-distance lift bounds of $28$ and $21$, respectively (Computational Theorem~\ref{thm:equal}). The same mechanism explains the low-weight cluster observed in the ISD census at $n=153$ (Section~\ref{sec:bimodal}). These results yield a practical \emph{ancestor sieve} (Algorithm~\ref{alg:sieve}).

2) \emph{A certification pipeline} (Section~\ref{sec:pipeline}). This is a three-stage method: ISD screening with early exit, exact Brouwer--Zimmermann computation on the additive dual, and independently re-verified witnesses. It records CPU time and provides scripts, logs, and witnesses for reproducibility. The same workflow also rejects archived candidates that fail the certification checks, illustrating the distinction between heuristic screening and exact certification.

3) \emph{A palindromic quasi-cyclic family} (Section~\ref{sec:palin}). For $(1,2)$-QC additive codes $\langle(g v_1, g v_2)\rangle$ over $R=\Ftwo[x]/(x^n-1)$, palindromic multipliers force the defect symbol $P(x)$, and hence the symplectic Gram matrix of all cyclic shifts, to vanish identically (Theorem~\ref{thm:pal}). The resulting additive code is therefore self-orthogonal in every length. A reversal-pairing analysis (Proposition~\ref{prop:rev}) shows that 1-generator QC codes are self-orthogonal under a conjugate form. An exact symplectic Gram rank formula and subset-sum constraints describe the possible gcd degrees, with complete realization verified for the small lengths examined (Theorem~\ref{thm:spectrum}).

4) {\sloppy\emph{Twenty-one certified record-breaking codes in twenty entries} (Section~\ref{sec:record}), as applications of the theory and the pipeline. The certified set includes $[[73,10,15]]_2$, $[[91,15,17]]_2$, $[[93,15,17]]_2$, $[[93,16,17]]_2$, and $[[60,13,12]]_2$, together with sixteen further codes listed in Table~\ref{tab:certified}. Eighteen are palindromic QC codes; the other three have non-palindromic multipliers but identically vanishing defect symbols, and two of these three come from the directed construction of Proposition~\ref{prop:directed}. Nineteen entries are improved by exactly one and the $[[60,13]]_2$ entry by two, with a second certified member at $d=11$. Propagation improves thirty-four further entries, including the exact-distance chains $[[n,15,17]]_2$ for $92\le n\le100$ and $[[n,14,17]]_2$ for $91\le n\le99$ (the \hyperref[rem:derived]{propagation} and \hyperref[rem:subcode]{subcode constructions}). We also give explicit constructions for the table-tying GV entries $[[101,15,17]]_2$, $[[100,14,17]]_2$, and $[[101,14,17]]_2$, all marked \texttt{Code Missing Construction}. Evidence-level candidates appear in Section~\ref{sec:candidates}.\par}

This paper is organized as follows. Section~\ref{sec:prelim} fixes notation. Sections~\ref{sec:palin}--\ref{sec:lift} develop the family theory, and Sections~\ref{sec:spectrum}--\ref{sec:assembly} present the rank and degree structure of the defect symbol, the directed construction, and the extension/assembly machinery. Section~\ref{sec:pipeline} describes certification, Sections~\ref{sec:record}--\ref{sec:candidates} present the codes, and Section~\ref{sec:concl} concludes with conjectures. Appendices give all construction data, witnesses, and artifacts.

\section{Preliminaries}\label{sec:prelim}

Let $\Ffour = \{0,1,w,w^2\}$, $w^2+w+1=0$. An additive code $\mathcal{C} \subseteq \Ffour^n$ is an $\Ftwo$-linear subspace; a generator $G = A + wB$ has binary symplectic representation $(A|B)$. The trace inner product $\omega(u,v) = \tr(u\bar v)$ is the trace-symplectic form on the binary pair representation. On binary pairs $(a|b)$ and $(c|d)$, it reads $a\cdot d + b\cdot c$. A self-orthogonal $\mathcal{C}$ with $t$ generators yields $[[n, n-t, d_Q]]_2$ with $d_Q = \min\swt(\mathcal{C}^{\perp_{ts}}\setminus\mathcal{C})$, where $\swt(a|b) = |\supp(a)\cup\supp(b)|$. The symplectic Gram matrix $S = AB^{\mathsf T} + BA^{\mathsf T}$ of the generator shifts is alternating of even rank $r$; we call $r$ the \emph{symplectic Gram rank}. Darboux standardization over $\Ftwo$ puts it in the form $\mathrm{diag}(J_2^{r/2},0)$ by pure row operations. We write $\mathcal{C}^{\perp_{ts}}$ for the dual with respect to $\omega$, and call a word of $\mathcal{C}^{\perp_{ts}}$ \emph{pure-$X$} if it has the form $(p \mid 0)$ and \emph{pure-$Z$} if it has the form $(0 \mid q)$. This is unrelated to the pure/impure dichotomy of \cite{crss}, which classifies a code by the minimum weight of its dual.

We work in $R = \Ftwo[x]/(x^n-1)$, with reversal involution $\rho(p)(x) = p(x^{-1})$; $p$ is \emph{palindromic} if $\rho(p) = p$. A $(1,2)$-QC additive code is the cyclic closure $\langle(a,b)\rangle = \{(s a, s b) : s \in R\}$. Best-known distances are taken from Grassl's table \cite{grassl}.

\section{The Palindromic Quasi-Cyclic Family}\label{sec:palin}

\begin{definition}
The \emph{palindromic QC family} at length $n$ consists of the codes $\mathcal{C} = \langle(g v_1, g v_2)\rangle \subseteq R^2$ with $g \mid x^n-1$ and $v_1, v_2 \in R$ palindromic.
\end{definition}

For a $(1,2)$-QC code with generator $(a,b)$, the Gram matrix of the cyclic shifts under $\omega$ is circulant with \emph{defect symbol}
\begin{equation}
P(x) = a(x)\,\rho(b)(x) + \rho(a)(x)\,b(x) \in R.
\end{equation}

\begin{theorem}\label{thm:pal}
If $v_1, v_2$ are palindromic then $P \equiv 0$; hence every member of the palindromic QC family is symplectic self-orthogonal, for every $n$ and every $g \mid x^n-1$.
\end{theorem}
\begin{proof}
$P = g\rho(g)\,(v_1 \rho(v_2) + \rho(v_1) v_2) = g\rho(g)\,(v_1 v_2 + v_1 v_2) = 0$ in characteristic $2$, since $\rho(v_i) = v_i$.
\end{proof}

\begin{proposition}\label{prop:dim}
$\dim_{\Ftwo} \langle(a,b)\rangle = n - \deg\gcd(a, b, x^n-1)$; in particular the family member has $k = \deg g$ whenever $\gcd(v_1, v_2, (x^n-1)/g) = 1$.
\end{proposition}
\begin{proof}
The annihilator of the generator is the ideal $(x^n-1)/\gcd(a,b,x^n-1)$; the map $s \mapsto (sa, sb)$ has kernel of dimension $\deg\gcd(a,b,x^n-1)$.
\end{proof}

\begin{proposition}\label{prop:rev}
Define $\omega_\rho((a|b),(c|d)) = a\cdot \rev(d) + \rev(b)\cdot c$. Then (i) every 1-generator QC code is automatically $\omega_\rho$-self-orthogonal; (ii) the map $T(a|b) = (a \mid \rev(b))$ is invertible and conjugates $\omega_\rho$ to $\omega$; hence an $\omega_\rho$-self-orthogonal additive code $\mathcal{C}$ gives a stabilizer code $T(\mathcal{C})$ of the same length and dimension. Distances are evaluated on $T(\mathcal{C})$ itself, since $T$ permutes only the $Z$-block and the union weight $\swt$ is not preserved in general.
\end{proposition}
\begin{proof}
(i) For $u = (a,b)$, $\omega_\rho(u, x^i u) = [x^{n-i}](\rho(a)\,b) + [x^{n-i}](a\,\rho(b)) = 0$ in characteristic $2$, and $\omega_\rho$ is bilinear and alternating. (ii) $\omega(T(a|b), T(c|d)) = a\cdot\rev(d) + \rev(b)\cdot c = \omega_\rho((a|b),(c|d))$, and invertibility is clear.
\end{proof}

\begin{example}\label{ex:60}
An archived $[[60,8,13\text{--}14]]_2$ QC record (generator data reproduced in the artifact package, Appendix~\ref{app:artifacts}) illustrates why self-orthogonality and the choice of symplectic form must be checked before any distance claim. The generator is not self-orthogonal under $\omega$, but it is self-orthogonal under $\omega_\rho$, as Proposition~\ref{prop:rev}(i) requires. Hence $T(\mathcal{C})$, rather than $\mathcal{C}$ itself, is the relevant stabilizer code. Our pipeline (Section~\ref{sec:pipeline}) produced independently re-verified logical operators of weights $12, 11, 10$; in particular, a weight-$10$ logical operator certifies $d_Q \le 10 < 12$ and invalidates the archived record, but does not by itself establish $d_Q = 10$. Pairing conventions and exact certification are both indispensable.
\end{example}

Consistently, in an even-length sweep ($n \in \{42,\dots,90\}$, $\deg g \in \{4,\dots,20\}$) all $9\,620$ randomly generated palindromic candidates passed the self-orthogonality check, leaving a pure distance hunt.

\begin{remark}\label{rem:outside}
The palindromic criterion is sufficient but far from necessary. The self-orthogonal QC codes underlying our evidence-level candidates $[[127,21,\le 21]]_2$ and $[[153,17,\le 27]]_2$ have reduced multipliers $a/g$, $b/g$ that are not palindromic, even up to cyclic shifts (verified exhaustively). They arise instead from the directed construction of Proposition~\ref{prop:directed}, which solves the defect equation $P \equiv 0$ directly.
\end{remark}

\begin{proposition}\label{prop:reversible}
For a palindromic-multiplier code $\mathcal{C} = \langle (g v_1, g v_2) \rangle$ with $\gcd$-exact $g$, the following are equivalent: (i) $\mathcal{C}$ is \emph{reversible}, i.e., closed under $x \mapsto x^{-1}$ (so its automorphism group contains the dihedral group $D_n$); (ii) $g$ is self-reciprocal up to a shift. Otherwise $\mathcal{C}$ pairs with a distinct reversal twin (Proposition~\ref{prop:rev}). Verified: $[[80,7,16]]_2$ is reversible; $[[73,10,15]]_2$ and $[[56,9,12]]_2$ are not.
\end{proposition}
\begin{proof}
Reversal sends $(g v_1, g v_2)$ to a shift of $(g^* v_1, g^* v_2)$ (palindromy of $v_i$), which lies in $\mathcal{C}$ iff $g^* \in (g)$ in $R_n$ (using the gcd-exactness of $g$). Since $g \mid x^n-1$ implies $g^* \mid x^n-1$ with $\deg g^* = \deg g$, the ideals $(g)$ and $(g^*)$ of $R_n$ have the same $\Ftwo$-dimension $n - \deg g$, even when $x^n-1$ has repeated factors; hence the containment $g^* \in (g)$ forces $(g^*) = (g)$, i.e., $g$ is self-reciprocal up to a unit, and palindromy reduces the unit to a shift.
\end{proof}

\subsection{Duality structure of 1-generator QC codes}\label{sec:duality}

The results of this subsection describe the symplectic dual of a 1-generator QC code exactly; they hold for general (not necessarily palindromic) generators.

\begin{proposition}\label{prop:purebch}Let $\mathcal{C} = \langle (a,b) \rangle$ be a 1-generator QC code at length $n$ odd. The pure-$X$ part $\{(p \mid 0) \in \mathcal{C}^{\perp_{ts}}\} = \operatorname{ann}(b^*) \times \{0\}$ and the pure-$Z$ part are cyclic codes whose defining sets are the complements of the root sets of $b^*$ and $a^*$; the BCH bound applies to each section separately. In particular, if both annihilator defining sets contain consecutive runs of length $\ge d_Q$, then \emph{every} minimum-weight logical operator is necessarily mixed.
\end{proposition}
\begin{proof}
$(p \mid 0) \in \mathcal{C}^{\perp_{ts}}$ iff $p \cdot b(x^{-1}) = 0$ in $R_n$; for odd $n$ the ring $R_n$ is semisimple, so this annihilator is the cyclic ideal whose roots are exactly the non-roots of $b^*$.
\end{proof}

\begin{theorem}\label{thm:compdual}
Let $n$ be odd, $\mathcal{C} = \langle(a,b)\rangle \subseteq R_n^2$, and let $x^n - 1 = \prod_j f_j$ be the irreducible factorization, so $R_n \cong \prod_j F_j$ (CRT, $F_j$ fields). Then $\mathcal{C}^{\perp_{ts}} = \prod_j K_j$ where, writing $\alpha_j = a^*|_{F_j}$, $\beta_j = b^*|_{F_j}$:
(i) if $\alpha_j, \beta_j \neq 0$: $K_j = \{(\alpha_j \beta_j^{-1} s,\ s) : s \in F_j\}$, a graph line;
(ii) if $\beta_j = 0 \neq \alpha_j$: $K_j = F_j \times \{0\}$ (pure-$X$);
(iii) if $\alpha_j = 0 \neq \beta_j$: $K_j = \{0\} \times F_j$ (pure-$Z$);
(iv) if $\alpha_j = \beta_j = 0$: $K_j = F_j^2$ (free).
\end{theorem}
\begin{proof}
Since $n$ is odd, $x^n - 1$ is squarefree and the CRT decomposition $R_n \cong \prod_j F_j$ is a ring isomorphism onto a product of fields; hence the defining equation $p\,b^* + q\,a^* = 0$ of $\mathcal{C}^{\perp_{ts}}$ holds in $R_n$ if and only if $p_j\,\beta_j + q_j\,\alpha_j = 0$ holds in every field component $F_j$. Over a field this is either void (case (iv)) or a single linear constraint: if both $\alpha_j, \beta_j$ are nonzero, then $p_j = q_j\,\alpha_j\beta_j^{-1}$ with $q_j$ free (case (i), the sign being irrelevant in characteristic $2$); if exactly one of $\alpha_j, \beta_j$ vanishes, that coordinate is free and the other is forced to zero (cases (ii), (iii)). Both the equation and the code split over the same CRT components, so $\mathcal{C}^{\perp_{ts}} = \prod_j K_j$.
\end{proof}

\begin{corollary}\label{cor:graph}
Suppose $a^*$ and $b^*$ vanish on exactly the same components: $\gcd(a^*, x^n{-}1) = \gcd(b^*, x^n{-}1) = \tilde g$ (\emph{root-balanced}). With $h = (x^n{-}1)/\tilde g$, $S = \Ftwo[x]/(h)$, and the unit $c := (a^*/\tilde g)(b^*/\tilde g)^{-1} \bmod h$:
\begin{equation}
\mathcal{C}^{\perp_{ts}} = \{ (c\,q + h r,\ q) : q \in R_n,\ r \in R_n \}, \qquad \dim \mathcal{C}^{\perp_{ts}} = n + \deg \tilde g,
\end{equation}
the direct sum of the \emph{unit-ratio graph code} $\langle(c,1)\rangle$ (a 1-generator QC code) and the pure-$X$ annihilator $(h) \times \{0\}$. Both single-component sections equal the cyclic ideal $(h)$, and
\begin{equation}
d(\mathcal{C}^{\perp_{ts}}) = \min\Big\{\, d\big((h)\big),\; \min_{q \neq 0}\; \min_{w \in (h)}\, \big|\supp(q) \cup \supp\big(c\,q + w\big)\big| \,\Big\},
\end{equation}
the first term BCH-bounded (Proposition~\ref{prop:purebch}), the second a coset-minimum problem over the single-component ideal: the inner minimum is essential, since the canonical representative $c\,q \bmod h$ is generally far from optimal. Verified exactly for $[[73,10,15]]_2$: $\deg \tilde g = 10$, $h$ squarefree of degree $63$, the displayed generating set has rank $83 = \dim \mathcal{C}^{\perp_{ts}}$, and $300$ random dual words satisfy the graph law. So does the certified weight-$15$ witness itself: its canonical graph representative has $X$-part Hamming weight $37$ (symplectic weight $44$), while its $(h)$-correction reduces the $X$-part to weight $11$. The minimum word is a \emph{deep coset leader}.
\end{corollary}

\begin{proposition}\label{prop:orderbound}
In the root-balanced setting of Corollary~\ref{cor:graph}, let $c \in S^\times$ have multiplicative order $t$ and put $e_t := \sum_{i=0}^{t-1} c^i$. If $e_t \neq 0$, then $c \cdot e_t = e_t$, and every $q = e_t s$ with $q \neq 0$ gives a ``diagonal'' word $(cq, q) = (q,q)$; consequently
\begin{equation}
d_{\mathrm{mix}} \;\le\; \min_{s:\, e_t s \neq 0}\, \wt_S(e_t\, s).
\end{equation}
If $e_t=0$, this construction gives no nonzero diagonal word and hence no finite bound. In fact, for any $c$ of order $t>1$, $e_t=0$, so the order-based bound is degenerate in every nontrivial case; only $c=1$ gives $e_t\neq0$ (with $t=1$). The componentwise analogue uses only those components with $c_j=1$.
\end{proposition}
\begin{proof}
$c \cdot e_t = \sum_{i=1}^{t} c^i = e_t$ since $c^t = 1$ in characteristic $2$; then $\swt(cq \mid q) = \wt(q)$. Moreover, for $t>1$, $e_t = (c^t+1)/(c+1) = 0$ because $c^t=1$ and $c\neq1$.
\end{proof}

\begin{theorem}\label{thm:localdual}
Write $n = 2^s n'$ with $n'$ odd and $x^{n'}-1 = \prod_j f_j$ the irreducible factorization, so $R_n \cong \prod_j R_j$ with $R_j := \Ftwo[x]/(f_j^{2^s})$ local Artinian. For $\mathcal{C} = \langle(a,b)\rangle$ write $a^*|_j = f_j^{\alpha_j} u_j$, $b^*|_j = f_j^{\beta_j} v_j$ with $u_j, v_j \in R_j^\times$; interchanging the two coordinates on $R_j$ if necessary, assume $\alpha_j \le \beta_j$. Then $\mathcal{C}^{\perp_{ts}} = \prod_j K_j$ with
\begin{equation}
K_j = \big\{ \big(p,\; c_j\, p + f_j^{\,2^s - \alpha_j} t\big) : p, t \in R_j \big\}, \qquad c_j := f_j^{\beta_j - \alpha_j} v_j u_j^{-1},
\end{equation}
a \emph{nilpotent-sheared graph} of dimension $(2^s + \alpha_j)\deg f_j$, the boundary cases included: $\beta_j = 2^s$ forces the second coordinate into the nilpotent ideal $(f_j^{\,2^s-\alpha_j})$, and $\alpha_j = \beta_j = 2^s$ leaves $K_j = R_j^2$ free. Summing over components,
\begin{equation}
\dim \mathcal{C}^{\perp_{ts}} = n + \sum_j \min(\alpha_j, \beta_j)\deg f_j = n + \deg\gcd(a, b, x^n{-}1).
\end{equation}
Theorem~\ref{thm:compdual} is the case $s = 0$; we call $\mathcal{C}$ \emph{valuation-balanced} when $\alpha_j = \beta_j$ for all $j$.
\end{theorem}
\begin{proof}
The CRT decomposition $R_n \cong \prod_j R_j$ (the $f_j^{2^s}$ pairwise coprime) reduces the defining equation $p\,b^* + q\,a^* = 0$ to each component $R_j$, where, since $\alpha_j \le \beta_j$, it reads $f_j^{\alpha_j}\big(p\, f_j^{\beta_j-\alpha_j} v_j + q\, u_j\big) = 0$; as $u_j$ is a unit and $\operatorname{ann}\big(f_j^{\alpha_j}\big) = \big(f_j^{\,2^s-\alpha_j}\big)$, this is exactly $q = c_j p + f_j^{\,2^s-\alpha_j} t$ with $p$ free and $t$ free modulo $f_j^{\alpha_j}$, giving $\dim K_j = (2^s + \alpha_j)\deg f_j$. The same derivation covers the boundary values verbatim: $\beta_j = 2^s$ gives $c_j \in \big(f_j^{\,2^s-\alpha_j}\big)$, and $\alpha_j = \beta_j = 2^s$ makes the equation vacuous. Summing $\big(2^s + \min(\alpha_j, \beta_j)\big)\deg f_j$ over components telescopes by $\deg\gcd(a, b, x^n{-}1) = \sum_j \min(\alpha_j, \beta_j)\deg f_j$ and $n = 2^s \sum_j \deg f_j$.
\end{proof}

\begin{remark}\label{rem:valuations}
The six even-length certified codes profiled here satisfy the dimension identity, and the sheared-graph law was verified on $200$ random dual words for three of them. Profiles (component: $(\alpha,\beta)$): $[[56,9,12]]_2$: $(x{+}1)^8$: $(5,3)$; $[[62,16,11]]_2$: $(x{+}1)^2$: $(1,2)$; $[[62,17,11]]_2$: valuation-balanced; $[[78,13,14]]_2$: deg-$2$ component $(1,0)$; $[[78,14,14]]_2$: deg-$2$ component $(0,1)$; $[[80,7,16]]_2$: $(x{+}1)^{16}$: $(3,5)$. Unbalanced components force one section of every dual word into a nilpotent ideal: every $[[80,7,16]]_2$ dual word has $Z$-part divisible by $(x{+}1)^2$, and every $[[56,9,12]]_2$ dual word has $X$-part divisible by $(x{+}1)^2$. This structural constraint is invisible at the block level; it is the local counterpart of the two-layer analysis of Section~\ref{sec:lift}.
\end{remark}

\begin{remark}\label{rem:maxord}
The order bound of Proposition~\ref{prop:orderbound} gives no useful upper bound for our best codes: for any $c$ of order $t>1$, the geometric sum $e_t$ vanishes, so the construction yields no nonzero diagonal word. In particular, $\operatorname{ord}(c) = 511 = 2^9-1$ for the certified $[[73,10,15]]_2$ and $\operatorname{ord}(c) = 1023 = 2^{10}-1$ for the certified root-balanced $[[93,16,17]]_2$, maximal in every component field of $S$. This observation concerns only that particular order-based bound and does not exclude other algebraic explanations. The exact distances in this paper are established by direct computation rather than by Proposition~\ref{prop:orderbound}. Root-balance is not automatic either: the certified $[[91,15,17]]_2$ and $[[93,15,17]]_2$ codes have unequal $a^*$, $b^*$ root sets ($\deg 15$ vs.\ $16$, and $18$ vs.\ $15$), so their duals additionally contain the single-component layers of Theorem~\ref{thm:compdual}(ii)--(iii).
\end{remark}

\begin{remark}\label{rem:73mixed}
For the certified $[[73,10,15]]_2$ code, $a^*$ and $b^*$ share the same $10$ roots, so the code is root-balanced. Both single-component sections are the cyclic $[73,10,28]_2$ code. Hence no single-component logical operator has weight below $28>d_Q=15$, and the certified witness is mixed. By Corollary~\ref{cor:graph}, certifying $d_Q$ is equivalent to computing the symplectic distance of the explicit unit-ratio graph code $\langle(c,1)\rangle$ over $S$ in dimension $63$. This example motivates studying symplectic BCH-type bounds for unit-ratio graph codes.
\end{remark}

\section{Ancestor Lifting}\label{sec:lift}

Throughout this section $n = \ell m$ with $m \mid n$. Write $\varphi: R_n \to R_m$, $p \mapsto p \bmod (x^m-1)$ (the \emph{fold}), and $L: R_m \to R_n$, $u \mapsto u \cdot e$ with $e = (x^n-1)/(x^m-1)$ (the \emph{lift}). As vectors, $\varphi$ folds the $\ell$ blocks and $L$ repeats $u$ periodically $\ell$ times.

\begin{remark}\label{rem:positioning}
The CRT decomposition of quasi-cyclic codes along the factors of $x^n - 1$ is classical for \emph{linear} codes over fields~\cite{lally,lingsolle}, and repeated-root cyclic codes ($\gcd(n,q) > 1$) were analyzed in~\cite{castagnoli,vanlint}. The fold $\varphi$ itself, blockwise summation modulo $x^m - 1$, is known from the quasi-cyclic literature~\cite{bcq12,ht15}. Our contribution concerns its action on the trace-symplectic dual (Theorem~\ref{thm:lift}) and the resulting exact distance law. Theorems~\ref{thm:twolayer} and~\ref{thm:crt} have three distinctive features. First, they apply to $\mathbb{F}_4$-\emph{additive} codes, and every certified example lies strictly outside the $\mathrm{GF}(4)$-linear class (the \hyperref[rem:gf4]{GF(4)-linearity discussion}), so the classical linear theory does not apply directly. Second, the $\ell=2$ analysis tracks the \emph{symplectic dual} rather than the code: an obstruction map $\widehat{R}$ couples the layers, a phenomenon absent in the odd-$\ell$ coprime-layer case, where the layers decouple by CRT. Third, the output is an \emph{exact} distance formula (Corollary~\ref{cor:dimlaw}) tied to the certification pipeline, not only a structural decomposition. In the closest recent work~\cite{guan-symp}, necessary and sufficient symplectic self-orthogonality conditions are established for \emph{linear} quasi-cyclic codes of even index, and record-breaking codes are derived from them; our certified codes are $\mathbb{F}_4$-additive and outside that linear class (the \hyperref[rem:gf4]{GF(4)-linearity discussion}), so the two families of constructions are complementary.
\end{remark}

\begin{theorem}\label{thm:lift}
Let $n = \ell m$, $\mathcal{C} = \langle(a,b)\rangle \subseteq R_n^2$, and $\mathcal{C}_m = \langle(\varphi a, \varphi b)\rangle \subseteq R_m^2$. Then:
\begin{enumerate}
\item $L$ is injective, $\swt(L u, L v) = \ell \cdot \swt(u,v)$, and $L$ intertwines the cyclic shift: $L(xu) = x\,L(u)$;
\item $(u,v) \in \mathcal{C}_m^{\perp_{ts}} \implies (Lu, Lv) \in \mathcal{C}^{\perp_{ts}}$;
\item $d(\mathcal{C}^{\perp_{ts}}) \le \ell \cdot d(\mathcal{C}_m^{\perp_{ts}})$, and if a minimum-weight logical operator of $\mathcal{C}_m$ lifts outside $\mathcal{C}$ then $d_Q(\mathcal{C}) \le \ell \cdot d_Q(\mathcal{C}_m)$.
\end{enumerate}
\end{theorem}
\begin{proof}
(1) $L(u)$ is the period-$m$ repetition of $u$ ($e = 1 + x^m + \cdots + x^{(\ell-1)m}$), so $L$ is injective; shift-intertwining is commutativity of multiplication in $R_n$; and the weight identity holds because the $\ell$ blocks are disjoint translates.
(2) Writing indices as $j = r + s m$, $r \in [0,m)$, $s \in [0,\ell)$,
\begin{align*}
\omega_n((Lu, Lv), x^t(a,b)) &= \textstyle\sum_j \big(L(u)_j b_{j+t} + L(v)_j a_{j+t}\big)\\
&= \textstyle\sum_r \big(u_r \sum_s b_{r+sm+t} + v_r \sum_s a_{r+sm+t}\big)\\
&= \omega_m\big((u,v),\, x^{t'}(\varphi a, \varphi b)\big) = 0
\end{align*}
for a suitable shift $t'$, since $(u,v)$ is orthogonal to every shift of $(\varphi a, \varphi b)$.
(3) Immediate from (1) and (2).
\end{proof}

\begin{corollary}\label{cor:div}
If $(x^m-1) \mid (\alpha a + \beta b)$ in $R_n$ for some $m \mid n$, $m<n$, and some $(\alpha,\beta) \in \Ftwo^2\setminus\{0\}$, then $d(\mathcal{C}^{\perp_{ts}}) \le n/m$; unless the corresponding lift lies in $\mathcal{C}$, $d_Q \le n/m$.
\end{corollary}
\begin{proof}
Write $e = (x^n-1)/(x^m-1)$. The three cases give weight-$n/m$ elements of $\mathcal{C}^{\perp_{ts}}$: $(\alpha,\beta)=(1,0)$: the pure-$Z$ words $x^j e$; $(0,1)$: the pure-$X$ words $x^j e$; $(1,1)$: the words $(x^j e, x^j e)$. In each case $\varphi$ of the relevant combination vanishes, so the weight-$1$ element $(\beta, \alpha)$ of $R_m^2$ lies in $\mathcal{C}_m^{\perp_{ts}}$; apply Theorem~\ref{thm:lift}.
\end{proof}

\begin{lemma}\label{lem:fold}
In the setting of Theorem~\ref{thm:lift}, $\varphi(\mathcal{C}^{\perp_{ts}}) \subseteq \mathcal{C}_m^{\perp_{ts}}$; folding preserves dual membership, and an ancestor logical operator is obtained only when the folded word lies outside $\mathcal{C}_m$. Moreover $\varphi \circ L = \ell \cdot \mathrm{id}$, which is the identity for odd $\ell$ and zero for even $\ell$.
\end{lemma}
\begin{proof}
For $(u,v) \in \mathcal{C}^{\perp_{ts}}$, expand $\omega_m((\varphi u, \varphi v), x^t(\varphi a, \varphi b))$ as $\sum_{\Delta} \omega_n((u,v), x^{t+\Delta m}(a,b)) = 0$, grouping the block indices by their difference $\Delta$ modulo $\ell$. The composition identity follows from $\varphi(e) = \ell$ in characteristic $2$.
\end{proof}

\begin{theorem}\label{thm:twolayer}
Let $n = 2m$, and write a word $w \in R_n^2$ as $w = u + e \cdot z$ with $u = \varphi(w)$ (the fold) and $z$ its second block. Then $w \in \mathcal{C}^{\perp_{ts}}$ if and only if
\begin{enumerate}
\item[(i)] \emph{Layer 0:} $u \in \mathcal{C}_m^{\perp_{ts}}$;
\item[(ii)] \emph{Layer 1:} $\Psi(z) = R_u$, where $\Psi(z) = z_a\,(\varphi b)^* + z_b\,(\varphi a)^*$ in $R_m$ (star denoting $x \mapsto x^{-1}$) and $R_u$ is the restriction to exponents $[0,m)$ of the length-$n$ correlation polynomial $u_a\, b(x^{-1}) + u_b\, a(x^{-1})$.
\end{enumerate}
Moreover $\ker \Psi = \mathcal{C}_m^{\perp_{ts}}$, so the admissible $z$ for a consistent $u$ form an affine coset $z_0 + \mathcal{C}_m^{\perp_{ts}}$; the symplectic weight has the Plotkin form $\swt(w) = \swt(z) + \swt(u{+}z)$; and the consistency condition $R_u \in \operatorname{im} \Psi$ is linear in $u$.
\end{theorem}
\begin{proof}
For a shift lag $s \in [0,m)$, the orthogonality equation at lag $s$ expands as $E_s = R_u(s) + \Psi(z)(s)$, separating the second-block unknowns $z$ from the fold $u$. The lag pair $(s, s{+}m)$ sums to the fold condition $\omega_m(u, x^s(\varphi a, \varphi b)) = 0$; once layer 0 holds, the equations at lags $s$ and $s{+}m$ coincide, so the $2m$ shift equations reduce to (i) and the $m$ equations (ii). The kernel claim is the definition of $\mathcal{C}_m^{\perp_{ts}}$, the weight identity is immediate from disjoint block supports, and linearity of $R_u$ in $u$ gives the last claim.
\end{proof}

\begin{corollary}\label{cor:dimlaw}
With the obstruction map $\widehat{R}\colon \mathcal{C}_m^{\perp_{ts}} \to R_m / \operatorname{im}\Psi$, $u \mapsto R_u \bmod \operatorname{im}\Psi$,
\begin{equation}
\dim \mathcal{C}^{\perp_{ts}} = 2\,\dim \mathcal{C}_m^{\perp_{ts}} - \rank \widehat{R},
\end{equation}
and
\begin{equation}
d(\mathcal{C}^{\perp_{ts}}) = \min\Big\{\, 2\, d(\mathcal{C}_m^{\perp_{ts}}),\ \min_{\substack{u \in \mathcal{C}_m^{\perp_{ts}} \setminus \{0\} \\ R_u \in \operatorname{im}\Psi}}\ \min_{z \in S(u)} \swt(z) + \swt(u{+}z) \,\Big\},
\end{equation}
where $S(u) = z_0 + \mathcal{C}_m^{\perp_{ts}}$ is the layer-1 coset. For the quantum distance, the restriction to $\mathcal{C}^{\perp_{ts}}\setminus\mathcal{C}$ gives
\begin{equation}
d_Q = \min\Bigg\{
\min_{\substack{z\in\mathcal C_m^{\perp_{ts}}\\Lz\notin\mathcal C}} 2\,\swt(z),\
\min_{\substack{u\in\mathcal C_m^{\perp_{ts}}\setminus\{0\}\\R_u\in\operatorname{im}\Psi}}\
\min_{\substack{z\in S(u)\\u+ez\notin\mathcal C}}\swt(z)+\swt(u+z)
\Bigg\},
\end{equation}
with empty branches omitted. The non-stabilizer condition is imposed on the reconstructed word $u+ez$, not on $u$ and $z$ separately. The ``escape'' phenomenon is a coset beating the diagonal $2\,\swt(u)$; coexistence is equality of the two branches. We call a divisor edge $m \mid n$ \emph{lift-saturated} when the lift branch alone attains the minimum, \emph{coexisting} when both branches attain, and \emph{escaping} when the coset branch lies strictly below $2\,d(\mathcal{C}_m^{\perp_{ts}})$. For odd $\ell$ we borrow the same terms, with the mixed branch of Theorem~\ref{thm:crt} in place of the coset branch.
\end{corollary}

\begin{remark}\label{rem:squarefree}
Writing $\sigma(g) = \deg g - \deg \gcd(\varphi a, \varphi b, x^m{-}1)$ (the \emph{fold shrink}, equal to the degree of the squareful part of $g$ relative to $x^m{-}1$), our data are consistent with the following one-directional rule: $\sigma(g) = 0$ (no repeated factor of $g$ relative to $x^m{-}1$) implies no strict escape, $d(\mathcal{C}^{\perp_{ts}}) = 2\,d(\mathcal{C}_m^{\perp_{ts}})$. The $\sigma = 0$ codes $[[56,9,12]]_2$, $[[80,7,16]]_2$, $[[78,13,14]]_2$ satisfy equality, the last with coexistence. The squareful examples show different deviations: $[[78,14,14]]_2$ ($\sigma = 1$, no escape), $[[62,16,11]]_2$ and $[[62,17,11]]_2$ ($\sigma = 5$ and $12$, one below the lift bound), $[[93,15,17]]_2$ ($\sigma = 10$, four below the dual-distance lift bound $21$), and $[[91,15,17]]_2$ ($\sigma = 15$, eleven below the dual-distance lift bound $28$). We conjecture the forward direction: $\sigma(g) = 0$ forces the coset branch to attain at least $2\,d(\mathcal{C}_m^{\perp_{ts}})$.
\end{remark}

\begin{remark}\label{rem:62verify}For the certified $[[62,16,11]]_2$ code and its $m = 31$ fold: $\dim \mathcal{C}_{31}^{\perp_{ts}} = 42$ (note $\deg\gcd(\varphi a, \varphi b, x^{31}{-}1) = 11 \ne \deg g$), $\dim \operatorname{im}\Psi = 20$ confirming $\ker\Psi = \mathcal{C}_{31}^{\perp_{ts}}$, $\rank \widehat{R} = 6$, and the dimension law gives $2\cdot 42 - 6 = 78 = \dim \mathcal{C}^{\perp_{ts}}$, matching the direct computation. The certified weight-$11$ witness decomposes as $u$ of weight $9$ in layer 0 and $z$ of weight $7$ in layer 1 with $\Psi(z) = R_u$ and $7 + 4 = 11 < 12 = 2\, d(\mathcal{C}_{31}^{\perp_{ts}})$, realizing the escape branch of Corollary~\ref{cor:dimlaw} exactly. The $[[56,9,12]]_2$ code at $m = 28$ gives $\dim \mathcal{C}_{28}^{\perp_{ts}} = 37$, $\rank \widehat{R} = 9$, and $2 \cdot 37 - 9 = 65 = \dim \mathcal{C}^{\perp_{ts}}$, again matching.
\end{remark}

\begin{theorem}\label{thm:crt}
Let $n = \ell m$ with $\ell$ odd, and put $e = (x^n-1)/(x^m-1) = 1 + x^m + \cdots + x^{(\ell-1)m} = \prod_{d \mid \ell,\, d > 1} \Phi_d(x^m)$ (for prime $\ell$ this is $e = \Phi_\ell(x^m)$). Then $x^{n}-1 = (x^m-1)\,e$ with coprime factors, so $R_n \cong R_m \oplus \widetilde{S}$ with $\widetilde{S} = \mathbb{F}_2[x]/(e)$, and the trace-symplectic dual splits as a \emph{direct product}
\begin{equation}
\mathcal{C}^{\perp_{ts}} \;\cong\; \mathcal{C}_m^{\perp_{ts}} \;\oplus\; K_{\widetilde{S}},
\end{equation}
with \emph{no compatibility condition between the layers}: layer 0 is the fold $u = w \bmod (x^m{-}1) \in \mathcal{C}_m^{\perp_{ts}}$, and layer 1 is the component $\zeta = w \bmod e$ satisfying the same correlation identity in $\widetilde{S}$. In particular $\dim \mathcal{C}^{\perp_{ts}} = \dim \mathcal{C}_m^{\perp_{ts}} + \dim K_{\widetilde{S}}$, and with the idempotents $e_1 = e$, $e_2 = 1 + e$ every dual word reconstructs as $w = u\, e_1 + \zeta\, e_2$, with symplectic weight given by the block sum of the reconstruction (for $\ell = 3$: blocks $u{+}\zeta_1$, $u{+}\zeta_0$, $u{+}\zeta_0{+}\zeta_1$, generalizing the Plotkin form). The $\zeta = 0$ branch is exactly the lift branch of weight $\ell \cdot \swt(u)$.
\end{theorem}
\begin{proof}
Since $x^m \equiv 1 \pmod{(x^m-1)}$ we have $e \equiv \ell \equiv 1 \pmod{(x^m-1)}$ for odd $\ell$, so $e$ is coprime to $x^m - 1$; this gives the CRT splitting without any semisimplicity assumption on $R_m$. The orthogonality condition is the polynomial identity $p\,b(x^{-1}) + q\,a(x^{-1}) = 0$ in $R_n$, which splits componentwise, and the component maps are surjective with independent kernels. In $R_n$ we have $e\,x^m = e$ (since $x^n = 1$), hence $e^2 = \ell\, e = e$, and the reconstruction and weight formulas follow from the idempotent expansion.
\end{proof}

\begin{remark}\label{rem:93verify}
For the certified $[[93,15,17]]_2$ code ($93 = 3 \cdot 31$), the odd-$\ell$ CRT splitting gives $\dim \mathcal{C}_{31}^{\perp_{ts}} = 36$ and $\dim K_{\widetilde{S}} = 72$, hence $36 + 72 = 108 = \dim \mathcal{C}^{\perp_{ts}}$. The weight-$17$ logical has fold weight $15$ and reconstructs through the $\ell = 3$ block formula. For the certified $[[91,15,17]]_2$ code ($91 = 7 \cdot 13$), the corresponding dimensions are $13$ and $93$, giving $106$. The length-$13$ fold is additive self-dual with $d(\mathcal{C}_{13}^{\perp_{ts}}) = 4$. A minimum-weight ancestor word lifts to a weight-$28$ element of $\mathcal{C}$, not to a logical operator, so the certified distance $17$ is carried by the mixed branch rather than by a simple lift.
\end{remark}

\begin{table}[h]
\centering
\caption{Dimension-law verifications. $\ell=2$: $\dim = 2\dim \mathcal{C}_m^{\perp_{ts}} - \rank\widehat{R}$ (Theorem~\ref{thm:twolayer}); odd $\ell$: $\dim = \dim \mathcal{C}_m^{\perp_{ts}} + \dim K_{\widetilde{S}}$ (Theorem~\ref{thm:crt}).}
\label{tab:dimlaw}
\begin{tabular}{ccccc}
\toprule
code & $m$ & $\dim \mathcal{C}_m^{\perp_{ts}}$ & layer-1 datum & $\dim \mathcal{C}^{\perp_{ts}}$ (predicted = actual) \\
\midrule
$[[62,16,11]]_2$ & 31 & 42 & $\rank\widehat{R} = 6$ & $84 - 6 = 78 = 78$ \\
$[[56,9,12]]_2$ & 28 & 37 & $\rank\widehat{R} = 9$ & $74 - 9 = 65 = 65$ \\
$[[93,15,17]]_2$ & 31 & 36 & $\dim K_{\widetilde{S}} = 72$ & $36 + 72 = 108 = 108$ \\
$[[91,15,17]]_2$ & 13 & 13 & $\dim K_{\widetilde{S}} = 93$ & $13 + 93 = 106 = 106$ \\
\bottomrule
\end{tabular}
\end{table}

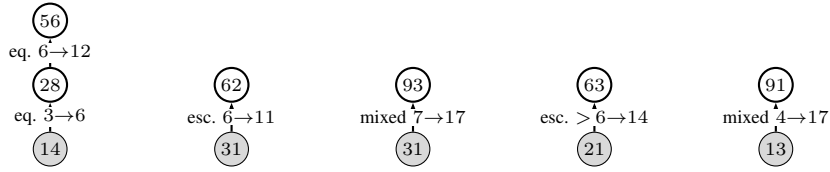
\begin{figure}[t]
\centering
\begin{tikzpicture}[
  font=\scriptsize,
  cert/.style={circle,draw,thick,inner sep=1.3pt},
  root/.style={circle,draw,fill=black!15,inner sep=1.3pt},
  ed/.style={-{Stealth[length=2mm]},thick},
  lbl/.style={midway,fill=white,inner sep=0.6pt}]
\node[root] (n14) at (0,0) {$14$};
\node[cert] (n28) at (0,0.85) {$28$};
\node[cert] (n56) at (0,1.7) {$56$};
\draw[ed] (n14) -- node[lbl]{eq.\ $3{\to}6$} (n28);
\draw[ed] (n28) -- node[lbl]{eq.\ $6{\to}12$} (n56);
\node[root] (n31) at (2.4,0) {$31$};
\node[cert] (n62) at (2.4,0.85) {$62$};
\draw[ed] (n31) -- node[lbl]{esc.\ $6{\to}11$} (n62);
\node[root] (n31b) at (4.8,0) {$31$};
\node[cert] (n93) at (4.8,0.85) {$93$};
\draw[ed] (n31b) -- node[lbl]{mixed\ $7{\to}17$} (n93);
\node[root] (n21) at (7.2,0) {$21$};
\node[cert] (n63) at (7.2,0.85) {$63$};
\draw[ed] (n21) -- node[lbl]{esc.\ ${>}\,6{\to}14$} (n63);
\node[root] (n13) at (9.6,0) {$13$};
\node[cert] (n91) at (9.6,0.85) {$91$};
\draw[ed] (n13) -- node[lbl]{mixed\ $4{\to}17$} (n91);
\end{tikzpicture}
\caption{Representative divisor-lattice chains for the certified codes. Edge labels give $d(\mathcal{C}_m^{\perp_{ts}}) \to d_Q(\mathcal C)$. The chains $14 \to 28 \to 56$ and $31 \to 62$ show equality and a one-unit deviation, respectively; $31 \to 93$ and $13 \to 91$ show mixed branches; $21 \to 63$ shows a deviation of at least seven. The self-dual fold at $13$ uses the exact dual distance $4$ rather than a quantum distance.}
\label{fig:lattice}
\end{figure}

\subsection{Lift-anchored construction}\label{sec:anchored}

The two-layer analysis inverts into a construction method: instead of analyzing a given code at $2m$, one \emph{prescribes} the fold and solves for the second layer.

\begin{definition}\label{def:anchored}
Let $(a,b)$ generate a self-orthogonal QC code at length $m$ (the \emph{anchor}). A length-$2m$ QC code $\mathcal{C}_2 = \langle (a + e\,u,\; b + e\,w) \rangle$ with $e = x^m + 1$ and $u, w \in R_m$ is called \emph{lift-anchored} at $(a,b)$; its fold is the anchor code.
\end{definition}

\begin{theorem}\label{thm:anchored}
Let the anchor be self-orthogonal, so its defect Laurent polynomial satisfies $D(x)\equiv 0$ in $R_m$; write $x^m D(x) = q \cdot (x^m - 1)$ with $q \in R_m$. Then $\mathcal{C}_2$ is self-orthogonal if and only if
\begin{equation}\label{eq:anchored}
T_b(u) + T_a(w) = q \quad \text{in } R_m, \qquad T_c(y) := c\, y^* + c^*\, y,
\end{equation}
with ${}^*$ denoting $x \mapsto x^{-1}$ in $R_m$. The admissible $(u,w)$ therefore form an affine subspace of $R_m^2$, nonempty iff $q \in \operatorname{im} T_a + \operatorname{im} T_b$; its dimension is $2m - \rank(T_a \mid T_b)$. Multiplying the anchor pair by a common unit preserves the anchor code and varies $q$, giving a twist family of consistency instances per anchor.
\end{theorem}
\begin{proof}
Put $e = x^m + 1$; in $R_{2m}$ we have $e^2 = x^{2m} - 1 \equiv 0$, $e^* = e$, and $\rho_{2m}(p) = \rho_m(p) + e\,(\rho_m(p) + p_0)$ for $p \in R_m$. Expanding the defect symbol of the generator $(a + e u,\; b + e w)$ of $\mathcal{C}_2$, the mixed terms $e^2 u\,\rho_{2m}(w)$ and $e^2 \rho_{2m}(u)\,w$ vanish, and reducing the $e$-coefficient modulo $e$ gives $P_2 = e\,\big(q + T_b(u) + T_a(w)\big)$ with $q$ as in the statement. Since $\operatorname{ann}_{R_{2m}}(e) = (e)$, the defect vanishes identically iff $T_b(u) + T_a(w) = q$ in $R_m$, which is \eqref{eq:anchored}.
\end{proof}

\begin{definition}\label{def:aperiodic}
An anchor is \emph{aperiodically self-orthogonal} when $D(x) \equiv 0$ as an unreduced polynomial (equivalently $q = 0$): self-orthogonality without the wrap-around. Then \eqref{eq:anchored} is homogeneous and the anchored codes form a \emph{linear} space.
\end{definition}

\begin{remark}\label{rem:anchoredstatus}
The $[[56,9,12]]_2$ anchor is aperiodically self-orthogonal. All seven certified codes serve as anchors, yielding $210$ lift-anchored candidates at $2n\in\{112,124,146,156,160\}$. Of the $56$ screened candidates, $55$ met the Grassl thresholds. In particular, the $[[124,17]]$ entry has a candidate with a weight-$20$ witness, matching the table value, and seven survivors. These candidates inherit $d_Q\le2d(\text{anchor})$; their exact certification requires further BZ computation beyond the present compute budget of Table~\ref{tab:cost}.
\end{remark}

\begin{remark}
Membership of a lift in $\mathcal{C}$ is decidable by linear algebra. In our non-self-dual ancestor examples, the minimum lifts are outside $\mathcal{C}$. A minimum-weight word of the self-dual length-$13$ fold of $[[91,15,17]]_2$ lifts into $\mathcal{C}$, so its weight-$28$ lift is not a logical operator. Theorem~\ref{thm:lift} is the additive-QC analogue of the classical $(u \mid u{+}v)$ structure of repeated-root cyclic codes \cite{castagnoli,vanlint}. The QC setting makes the folded ancestor itself a 1-generator QC code, so the mechanism iterates along the divisor lattice of $n$. Lifted and primitive minimum-weight logicals can coexist: in $[[78,13,14]]_2$ the ancestor satisfies $d_Q(39) = 7$, and we verified explicitly both a lifted weight-$14$ logical and a primitive weight-$14$ logical.
\end{remark}

\begin{comptheorem}\label{thm:equal}
For the codes listed in Table~\ref{tab:ancestor}, the relevant ancestor distances were computed exactly (BZ, each under $0.02$\,s). For the self-dual length-$13$ fold of $[[91,15,17]]_2$, the ancestor datum is $d(\mathcal{C}_{13}^{\perp_{ts}})$ rather than a quantum distance. Nine of these codes satisfy
\begin{equation}
d_Q(n) = \frac{n}{m}\, d_Q(m)
\end{equation}
at \emph{every} level of the divisor chain, with explicit lifted witnesses listed in Appendix~\ref{app:witnesses}. Seven further codes have a certified minimum distance exactly one below the corresponding ancestor-lifting bound: the two $[[62,\cdot,11]]_2$ and four $[[60,\cdot,11]]_2$ codes have $m \in \{30, 31\}$ ancestors with $d_Q = 6$, so direct ancestor lifting gives logical operators of weight at least $12$, whereas primitive minimum-weight logical operators have weight $11$. For $[[60,8,13]]_2$, the length-$30$ ancestor has quantum distance $7$, giving a lift bound of $14$, while the certified minimum distance is $13$. These minimum-weight logical operators arise from the mixed or coset branch rather than simple periodic lifting; in particular, the weight-$11$ witnesses of the $[[62,\cdot,11]]_2$ pair fold to ancestor logical operators of weight $9$ (Lemma~\ref{lem:fold}) and are therefore not disguised lifts. The extremal case is $[[63,7,14]]_2$, whose certified distance is at least seven below every divisor-lattice lift bound: every divisor of $63$ yields a lift bound $\ge 21$ (exhaustive coset-minimum search gives $d_Q(21) > 6$, and the divisibility penalty $(x^3-1) \mid (a{+}b)$ of Corollary~\ref{cor:div} yields only weight-$21$ words), and its minimum witness is primitive. The certified $[[91,15,17]]_2$ and $[[93,15,17]]_2$ codes attain $d_Q=17$ in the mixed branch, below the dual-distance lift bounds $7\cdot4=28$ and $3\cdot7=21$, respectively, and their minimum witnesses are primitive. The remaining code $[[73,10,15]]_2$ has prime length and is primitive by necessity.
\end{comptheorem}

\begin{table}[h]
\centering
\caption{Ancestor structure of selected certified codes. The ancestor-distance column gives $d_Q(m)$, except for the self-dual length-$13$ fold of $[[91,15,17]]_2$, where it gives $d(\mathcal{C}_{13}^{\perp_{ts}})$.}
\label{tab:ancestor}
\begin{tabular}{ccccc}
\toprule
Code & fold chain & ancestor distance & lift bound & certified $d_Q(n)$ \\
\midrule
$[[73,10,15]]$  & (73 prime) & --- & --- & 15 (primitive) \\
$[[56,9,12]]$   & $56{\to}28{\to}14$ & 3 & $4\cdot 3 = 12$ & 12 (equality) \\
$[[62,16,11]]$  & $62{\to}31$ & 6 & $2\cdot 6 = 12$ & 11 (escape by 1) \\
$[[62,17,11]]$  & $62{\to}31$ & 6 & $2\cdot 6 = 12$ & 11 (escape by 1) \\
$[[78,13,14]]$  & $78{\to}39$ & 7 & $2\cdot 7 = 14$ & 14 (equality; lift + primitive) \\
$[[78,14,14]]$  & $78{\to}39$ & 7 & $2\cdot 7 = 14$ & 14 (equality) \\
$[[80,7,16]]$   & $80{\to}40{\to}20$ & 4 & $4\cdot 4 = 16$ & 16 (equality) \\
\midrule
$[[60,8,13]]$   & $60{\to}30$ & 7 & $2\cdot 7 = 14$ & 13 (escape by 1) \\
$[[60,12,12]]$  & $60{\to}30$ & 6 & $2\cdot 6 = 12$ & 12 (equality) \\
$[[60,13,11]]$  & $60{\to}30$ & 6 & $2\cdot 6 = 12$ & 11 (escape by 1) \\
$[[60,13,12]]$  & $60{\to}30{\to}15{\to}5$ & 1 & $12\cdot 1 = 12$ & 12 (equality, every level) \\
$[[60,14,11]]$  & $60{\to}30$ & 6 & $2\cdot 6 = 12$ & 11 (escape by 1) \\
$[[60,15,11]]$  & $60{\to}30$ & 6 & $2\cdot 6 = 12$ & 11 (escape by 1) \\
$[[60,16,11]]$  & $60{\to}30$ & 6 & $2\cdot 6 = 12$ & 11 (escape by 1) \\
$[[60,19,10]]$  & $60{\to}30$ & 5 & $2\cdot 5 = 10$ & 10 (equality) \\
$[[60,20,10]]$  & $60{\to}30$ & 5 & $2\cdot 5 = 10$ & 10 (equality) \\
$[[62,7,14]]$   & $62{\to}31$ & 7 & $2\cdot 7 = 14$ & 14 (equality) \\
$[[63,7,14]]$   & $63{\to}21$ & ${>}\,6$ & ${\ge}\,3\cdot 7 = 21$ & 14 (escape by ${\ge}\,7$) \\
$[[91,15,17]]$   & $91{\to}13$ & 4 (self-dual fold) & $7\cdot 4 = 28$ & 17 (mixed branch) \\
$[[93,15,17]]$   & $93{\to}31$ & 7 & $3\cdot 7 = 21$ & 17 (mixed branch) \\
\bottomrule
\end{tabular}
\end{table}

\begin{algorithm}[h]
\caption{Ancestor sieve (design rule from Theorem~\ref{thm:lift})}
\label{alg:sieve}
\begin{algorithmic}[1]
\REQUIRE candidate $(g, v_1, v_2)$ at composite $n$; record target $d^*$
\STATE $B \leftarrow +\infty$
\FOR{each divisor $m \mid n$, $m < n$}
  \STATE compute the exact minimum structure of the folded ancestor
  \FOR{each verified minimum logical $w_m \in \mathcal{C}_m^{\perp_{ts}}\setminus\mathcal{C}_m$}
    \IF{$L(w_m) \notin \mathcal{C}$}
      \STATE $B \leftarrow \min(B,\, \swt(L(w_m)))$
    \ENDIF
  \ENDFOR
  \FOR{each divisibility witness $v$ of weight $n/m$}
    \IF{$v \notin \mathcal{C}$}
      \STATE $B \leftarrow \min(B,\, \swt(v))$
    \ENDIF
  \ENDFOR
\ENDFOR
\IF{$B \le d^*$} \STATE reject (a certified logical upper bound excludes the target) \ELSE \STATE pass to the ISD/BZ pipeline with no conclusion from the sieve \ENDIF
\end{algorithmic}
\end{algorithm}

\subsection{Structural explanation of the low-weight cluster in the $n=153$ ISD census}\label{sec:bimodal}

At $n = 153 = 3^2\cdot 17$, equal-budget ISD (200\,s) on 30 random family members with $\dim \mathcal{C} = 137$ produced two clusters of best-found values: 12 members at $9$--$18$ and 18 members at $31$--$35$. All 12 members in the low cluster have explicit lifted witnesses: $153 = 9\cdot 17 = 3\cdot 51$, and
\begin{itemize}
\item $(x^{17}-1) \mid b$: pure-$X$ lifts of weight $9$ ($3$ members; ISD best exactly $9$);
\item $(x^{9}-1) \mid b$: pure-$X$ lifts of weight $17$ ($4$ members; ISD best exactly $17$);
\item $(x^{9}-1) \mid (a+b)$: $(1,1)$-type lifts of weight $17$ ($2$ members; ISD best exactly $17$);
\item small $d_Q$ of the $m=51$ ancestor: lifts of weight $3\,d_Q(51)$ ($3$ members with witnesses: $18 = 3\cdot 6$, $15 = 3\cdot 5$, $12 = 3\cdot 4$).
\end{itemize}
Every low-cluster witness we examined is an exact period-$9$, period-$17$, or period-$51$ repetition (Appendix~\ref{app:artifacts}). The low-weight cluster therefore has a divisor-lattice explanation. The high best-found cluster is associated with members lacking an exploitable fold; its exact distances remain certification tasks.

\section{Rank and Degree Structure of the Defect Symbol}\label{sec:spectrum}

\begin{lemma}\label{lem:circ}
Let $\mathcal{C} = \langle(a,b)\rangle \subseteq R^2$ with shift basis $g_i = (x^i a, x^i b)$, $0 \le i < n$. The Gram matrix $S_{ij} = \omega(g_i, g_j)$ is the circulant matrix of the defect symbol: $S_{ij} = [x^{i-j}]\,P(x)$.
\end{lemma}
\begin{proof}
$\sum_k (x^i a)_k (x^j b)_k = \sum_k a_{k-i} b_{k-j} = \sum_l a_l b_{l + (i-j)} = [x^{i-j}](\rho(a)\,b)$ by the definition of the reversal involution, and similarly the second summand gives $[x^{i-j}](a\,\rho(b))$.
\end{proof}

\begin{theorem}\label{thm:spectrum}
For a $(1,2)$-QC additive code with defect symbol $P$:
\begin{enumerate}
\item the symplectic Gram rank is $r = n - \deg\gcd(P, x^n-1)$;
\item $P$ is palindromic and satisfies $P(1) = 0$, so $D := \gcd(P, x^n-1)$ is a self-reciprocal divisor of $x^n-1$ that necessarily contains $x+1$. For odd $n$, writing $x^n - 1 = (x+1)\,Q(x)$, the achievable gcd degrees are contained in $\{\,1 + s : s \in \Sigma(Q)\,\}$, where $\Sigma(Q)$ is the set of subset sums of the self-reciprocal block degrees of $Q$ $($a block is a self-reciprocal irreducible factor of $Q$, or a reciprocal pair $\{f, \rho(f)\}$ of combined degree $2\deg f)$; equality holds in every case we have computed (complete enumeration at $n = 5, 7, 9$ and directed search at $n = 15, 21$), and we conjecture that it holds for all odd $n$;
\item $r = 0 \iff P \equiv 0$.
\end{enumerate}
\end{theorem}
\begin{proof}
(1) By Lemma~\ref{lem:circ}, $S$ is the circulant matrix with symbol $P$. The right kernel of a circulant with symbol $P$ consists of the polynomials $v$ with $P v \equiv 0 \pmod{x^n-1}$, i.e., the ideal generated by $(x^n-1)/D$, where $D = \gcd(P, x^n-1)$; its dimension is $\deg D$, whence $r = n - \deg D$.
(2) $\rho(P) = \rho(a)\,b + a\,\rho(b) = P$, so $P$ is palindromic, hence $D = \gcd(P, x^n-1)$ is a self-reciprocal divisor of $x^n-1$. Moreover $P(1) = a(1)\,b(1) + a(1)\,b(1) = 0$ in characteristic $2$, so $(x+1) \mid D$; in particular $r \le n - 1$. For odd $n$ the factorization is squarefree, a self-reciprocal divisor containing $x+1$ is exactly $(x+1)D'$ with $D'$ a self-reciprocal divisor of $Q$, and the groupings of irreducible factors of $Q$ into self-reciprocal singletons and reciprocal pairs show that $\deg D'$ is a subset sum of the block degrees of $Q$. Whether every such divisor is attained in general remains open; complete enumeration of all $2^{2n}$ generator pairs at $n = 5, 7, 9$ gives achieved degree sets $\{1,5\}$, $\{1,7\}$, $\{1,3,7,9\}$ respectively, matching $\{1+s : s \in \Sigma(Q)\}$ in each case, and directed search attains every self-reciprocal divisor containing $x+1$ at $n = 15$ and $n = 21$. The same forced factor $x+1$ persists for even $n$ with multiplicities. (3) $r = 0 \iff S = 0 \iff$ every coefficient of $P$ vanishes.
\end{proof}

\begin{proposition}\label{prop:directed}
Fix $b \in R$ and a target self-reciprocal divisor $\mathrm{hr} \mid x^n-1$. The map $T_b: a \mapsto a\,\rho(b) + \rho(a)\,b$ is $\Ftwo$-linear, and the admissible set $\{a = g f_1 : \mathrm{hr} \mid T_b(a)\}$ is an $\Ftwo$-subspace of dimension $\ge (n - \deg g) - \deg \mathrm{hr}$.
\end{proposition}
\begin{proof}
Linearity is immediate from the linearity of $\rho$ and of multiplication by fixed elements. The admissible set is the kernel of the linear map $a \mapsto T_b(a) \bmod \mathrm{hr}$ from the $(n-\deg g)$-dimensional space of multiples of $g$ to the $\deg \mathrm{hr}$-dimensional quotient $R/(\mathrm{hr})$.
\end{proof}

\begin{remark}
When the dimension bound is positive, Proposition~\ref{prop:directed} constructs the admissible set by linear algebra rather than exhaustive random search. In the $n=21$, $r=2$ experiment, this produced prescribed-defect ingredients in all 20 trials, compared with 0 successes in 400 uniform random trials.
\end{remark}

\section{Symplectic Extension and Assembly}\label{sec:assembly}

\begin{theorem}\label{thm:ext}
An additive code $\mathcal{C}$ with symplectic Gram rank $r$ extends to a self-orthogonal additive code of length $n + r/2$ with the same number of generators, and $r/2$ new coordinates are necessary.
\end{theorem}
\begin{proof}
By Darboux standardization of the symplectic form on the generator space, the generator space decomposes as $V = V_0 \oplus \bigoplus_{i=1}^{r/2} H_i$ with $V_0$ isotropic and $H_i = \mathrm{span}(u_i, v_i)$ hyperbolic, $\omega(u_i, v_i) = 1$. Append $r/2$ coordinates and extend
\begin{equation}
u_i \mapsto u_i' = (u_i \mid z_i), \qquad v_i \mapsto v_i' = (v_i \mid x_i), \qquad V_0 \ni w \mapsto (w \mid 0),
\end{equation}
where $z_i$ (resp.\ $x_i$) is the unit vector with a single $Z$-entry (resp.\ $X$-entry) in the $i$-th new coordinate. The new coordinates contribute $z_i \cdot x_j = \delta_{ij}$ to the pairing of $u_i'$ and $v_j'$ and nothing to any other pairing, so $\omega(u_i', v_j') = \delta_{ij} + \delta_{ij} = 0$ in characteristic $2$, and the extended code is self-orthogonal. Necessity: if $c$ added coordinates make the code self-orthogonal, their contribution to the pairing of any two generators is an alternating form of rank at most $2c$ that must equal the original alternating form of rank $r$; hence $2c \ge r$.
\end{proof}

\begin{remark}\label{rem:ea}
An ingredient with symplectic Gram rank $r$ used in Theorem~\ref{thm:ext} also gives an entanglement-assisted quantum code $[[n,\, n - t + r/2,\, d;\, r/2]]_2$ under the Brun--Devetak--Hsieh construction~\cite{bdh}, with the symplectic Gram rank equal to the ebit count. Every directed-construction finalist of the sweeps has $r = 2$ and therefore has $c = 1$. The extension of Theorem~\ref{thm:ext} recovers the $r/2$ ebits at the price of $r/2$ additional physical coordinates. For example, the $[[34,11,6]]_2$ ingredient (Appendix~\ref{app:candidates}, full QC span) is an entanglement-assisted $[[33,11,7;1]]_2$, BZ-verified with $\min\swt(\mathcal{C}^{\perp_{ts}}) = 7$ and minimum words outside $\mathcal{C}$. One ebit increases the distance from $6$ to $7$ at one coordinate fewer. Against the entanglement-assisted table~\cite{grassl} (snapshot 2026-09-26), this certified $[[33,11,7;1]]_2$ is a strict improvement at $(n,k) = (33,11)$: the table lists no $c = 1$ entry at those parameters, and its distance-$7$ construction uses $c = 2$ ebits. A systematic comparison of the sweep's archive of ingredients with symplectic Gram rank $2$ against the entanglement-assisted table would come at no additional cost.
\end{remark}

\begin{theorem}\label{thm:assembly}
Let $\mathcal{C}_1, \mathcal{C}_2$ be additive codes of lengths $n_1, n_2$ with $t_1, t_2$ generators, both with symplectic Gram rank $r$. After Darboux standardization of both, the code generated by
\begin{equation}
(V_{0,1} \mid 0),\quad (0 \mid V_{0,2}),\quad (u_{1,i} \mid u_{2,i}),\quad (v_{1,i} \mid v_{2,i}),\quad 1 \le i \le r/2,
\end{equation}
is self-orthogonal with $t_1 + t_2 - r$ generators, giving a stabilizer code of length $n_1+n_2$ and dimension $n_1+n_2-(t_1+t_2-r)$ (in characteristic $2$ no twist factor is needed since $1+1=0$). Moreover, writing $w_i(s)$ for the minimum weight of a side-$i$ word with defect syndrome $s \in \Ftwo^r$,
\begin{equation}
d(\mathcal{C}^{\perp_{ts}}) = \min\Big\{d(\mathcal{C}_1^{\perp_{ts}}),\, d(\mathcal{C}_2^{\perp_{ts}}),\, \min_{s \ne 0}\big(w_1(s) + w_2(s)\big)\Big\}.
\end{equation}
\end{theorem}
\begin{proof}
Self-orthogonality: isotropic blocks pair trivially; $\omega((u_{1,i}|u_{2,i}), (v_{1,j}|v_{2,j})) = \delta_{ij} + \delta_{ij} = 0$; mixed pairings vanish because $V_{0,k}$ is orthogonal to everything in $\mathcal{C}_k$; the dimension count is $(t_1 - r) + (t_2 - r) + r$, with independence from the disjoint support structure. For the distance law, $(w_1 \mid w_2) \in \mathcal{C}^{\perp_{ts}}$ iff each side is orthogonal to its isotropic part and the two defect syndromes agree; $s = 0$ decouples the sides, and $s \neq 0$ adds the two minima over disjoint supports.
\end{proof}

\begin{remark}\label{rem:tinv}
At $r = 2$, all six identifications $M \in \Sp(2,2)$ between the hyperbolic parts left $d_Q$ unchanged in the tested instances. This is a numerical observation only; its general validity is open (Section~\ref{sec:concl}).
\end{remark}

\section{The Certification Pipeline}\label{sec:pipeline}

\begin{figure}[t]
\centering
\begin{tikzpicture}[
  font=\scriptsize,
  stage/.style={rectangle,draw,thick,rounded corners=1.5pt,align=center,minimum height=5.2mm,inner sep=2.2pt},
  arr/.style={-{Stealth[length=2mm]},thick},
  note/.style={font=\tiny,align=center}]
\node[stage,fill=black!8] (gen) at (0,0) {QC candidate\\generation};
\node[stage] (chk) at (2.5,0) {self-orthogonality\\check ($P{\equiv}0$)};
\node[stage] (isd) at (5.0,0) {ISD screen\\(early exit)};
\node[stage] (bz) at (7.5,0) {BZ exact\\(computation)};
\node[stage,fill=black!8] (wit) at (7.5,-1.5) {witness\\re-verify};
\node[stage,fill=black!15] (rec) at (5.0,-1.5) {certified\\code};
\node[stage,dashed] (evi) at (2.5,-1.5) {evidence-level\\candidate};
\draw[arr] (gen) -- (chk);
\draw[arr] (chk) -- (isd);
\draw[arr] (isd) -- node[midway,above,note]{survivors} (bz);
\draw[arr] (bz) -- (wit);
\draw[arr] (wit) -- (rec);
\draw[arr,dashed] (isd) -- node[midway,left,note]{witness; BZ\\not completed} (evi);
\end{tikzpicture}
\caption{Three-stage certification pipeline. Solid path: QC candidate generation $\to$ self-orthogonality filter $\to$ ISD screening with early exit $\to$ exact Brouwer--Zimmermann computation $\to$ independent witness re-verification. The output is a certified code; record status is assessed separately against Grassl's code table. The dashed exit applies only to candidates with explicit witnesses but no completed BZ certificate, which are reported at evidence level only (Remark~\ref{rem:evidence}).}
\label{fig:pipeline}
\end{figure}
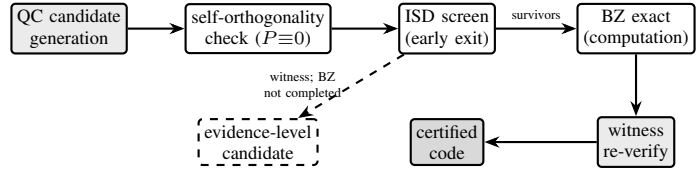

\emph{Stage 1: ISD screening with early exit.} Lee--Brickell information-set decoding is applied to the additive $\Ffour$-representation of $\mathcal{C}^{\perp_{ts}}$, in which a word $(a \mid b)$ is read as the $\Ffour$-vector $a + w b$, so that the symbol weight of a word is exactly its symplectic weight $\swt(a \mid b)$. Each screening run records its time budget and early-exit threshold. ISD yields only upper bounds; shallow screens can overestimate them, as illustrated by an archived candidate refined from $38$ to $19$.

\emph{Stage 2: exact certification.} Magma's Brouwer--Zimmermann \texttt{MinimumWeight} on the additive code $\mathcal{C}^{\perp_{ts}}$ proves $d_Q \ge d(\mathcal{C}^{\perp_{ts}})$. Table~\ref{tab:cost} records measured runtimes and the compute budget.

\emph{Stage 3: witness protocol.} Every claimed distance carries an explicit vector $W$: $W \in \mathcal{C}^{\perp_{ts}}$ and $W \notin \mathcal{C}$ re-verified by an independent linear-algebra computation (Appendix~\ref{app:witnesses}).

\begin{table}[h]
\centering
\caption{Measured exact-certification cost. CPU time is the Magma CPU time of the exact Brouwer--Zimmermann computation on the additive dual.}
\label{tab:cost}
\begin{tabular}{ccccc}
\toprule
Code & $\dim \mathcal{C}^{\perp_{ts}}$ & target $d$ & CPU time & status \\
\midrule
$[[56,9,12]]$  & 65 & 12 & 8.6\,s & certified \\
$[[62,16,11]]$ & 78 & 11 & 29\,s & certified \\
$[[62,17,11]]$ & 79 & 11 & 30\,s & certified \\
$[[73,10,15]]$ & 83 & 15 & 11\,478\,s & certified \\
$[[78,13,14]]$ & 91 & 14 & 1\,584\,s & certified \\
$[[78,14,14]]$ & 92 & 14 & 1\,699\,s & certified \\
$[[80,7,16]]$  & 87 & 16 & 14\,976\,s & certified \\
$[[80,7,16]]$ (member 2) & 87 & 16 & 19\,041\,s & certified \\
$[[91,15,17]]$ & 106 & 17 & 1\,625\,244.2\,s & certified \\
$[[93,15,17]]$ & 108 & 17 & 1\,850\,604.4\,s & certified \\
$[[93,16,17]]$ & 109 & 17 & 1\,996\,974.5\,s & certified \\
\midrule
$[[60,8,13]]$  & 68 & 13 & 183.3\,s & certified \\
$[[60,12,12]]$ & 72 & 12 & 20.9\,s & certified \\
$[[60,13,11]]$ & 73 & 11 & 22.3\,s & certified \\
$[[60,13,12]]$ & 73 & 12 & 23.4\,s & certified \\
$[[60,14,11]]$ & 74 & 11 & 25.0\,s & certified \\
$[[60,15,11]]$ & 75 & 11 & 25.2\,s & certified \\
$[[60,16,11]]$ & 76 & 11 & 31.0\,s & certified \\
$[[60,19,10]]$ & 79 & 10 & 56.5\,s & certified \\
$[[60,20,10]]$ & 80 & 10 & 505.1\,s & certified \\
$[[62,7,14]]$  & 69 & 14 & 1\,937.0\,s & certified \\
$[[63,7,14]]$  & 70 & 14 & 228.5\,s & certified \\
\midrule
$[[72,6]]$ (refuted) & 78 & 12 & 15\,201\,s & certified $d_Q = 12$; no record (Rem.~\ref{rem:72candidate}) \\
$[[127,21,\cdot]]$ & 148 & $\ge 20$ & not estimated & not completed within current compute budget \\
\bottomrule
\end{tabular}
\end{table}

\emph{Refutation case studies.} The pipeline also rejects unsupported archived candidates. An archived $[[70,24]]$ example gives $d=10$ rather than the ISD-suggested $11$ (source record \texttt{rec\_ext\_70\_24.txt}). The $[[72,6,16]]_2$ candidate is structurally invalid as recorded, and its legitimate $T$-image has certified $d_Q=12$ (Remark~\ref{rem:72candidate}); see also Example~\ref{ex:60}.

\begin{remark}\label{rem:env}
All exact certifications use Magma V2.25-3 \texttt{MinimumWeight} (Brouwer--Zimmermann) applied to the additive dual. The timings of Table~\ref{tab:cost} are Magma CPU times on an Intel Xeon Gold 5218R ($2.10$\,GHz, 128\,GB RAM). Generator matrices, witnesses, scripts, and full BZ logs are archived with SHA-256 manifests (Appendix~\ref{app:artifacts}), so each certified claim can be rerun or checked against the logs.
\end{remark}

\begin{remark}\label{rem:91cert}
For $[[91,15,17]]_2$, $\dim\mathcal{C}=76$ and $\dim\mathcal{C}^{\perp_{ts}}=106$. The archived weight-$17$ witness and Magma's returned minimum word both lie in $\mathcal{C}^{\perp_{ts}}\setminus\mathcal{C}$; hence the BZ lower bound and the witness upper bound coincide at $d_Q=17$ (\texttt{91\_15\_cert.m}, \texttt{91\_15\_cert.log}).
\end{remark}

\begin{remark}\label{rem:93cert}
For $[[93,15,17]]_2$, the defect symbol vanishes identically, $\dim\mathcal{C}=78$, and $\dim\mathcal{C}^{\perp_{ts}}=108$. The weight-$17$ witness in Appendix~\ref{app:witnesses} belongs to $\mathcal{C}^{\perp_{ts}}\setminus\mathcal{C}$, as independently re-verified by the orthogonality equation and rank augmentation. Magma's \texttt{MinimumWeight} returns $d(\mathcal{C}^{\perp_{ts}})=17$ after $1\,850\,604.391$\,s CPU, and the returned minimum word also lies outside $\mathcal{C}$. An independent computation with \texttt{cert\_odd\_4441} again returns $17$ after $1\,961\,291.640$\,s CPU. The latter is an independent check of the same code, not an additional certified member. The cumulative time above uses the primary run, with the independent rerun reported separately. The scripts and logs are archived as \texttt{93\_15\_cert.m}, \texttt{93\_15\_cert.log}, \texttt{cert\_odd\_4441.m}, and \texttt{cert\_odd\_4441.log}.
\end{remark}

\emph{Family calibration.} Equal-budget Stage-1 runs on random same-family members quantify selection effects (Section~\ref{sec:bimodal}; at $n=127$ the family is unimodal at best-found $22$--$25$).

\section{The Certified Record-Breaking Codes}\label{sec:record}

\begin{comptheorem}\label{thm:seven}
{\sloppy Each of $[[73,10,15]]_2$, $[[56,9,12]]_2$, $[[62,16,11]]_2$, $[[62,17,11]]_2$, $[[78,13,14]]_2$, $[[78,14,14]]_2$, $[[80,7,16]]_2$, $[[91,15,17]]_2$, $[[93,15,17]]_2$, $[[93,16,17]]_2$, $[[60,8,13]]_2$, $[[60,12,12]]_2$, $[[60,13,11]]_2$, $[[60,13,12]]_2$, $[[60,14,11]]_2$, $[[60,15,11]]_2$, $[[60,16,11]]_2$, $[[62,7,14]]_2$, $[[63,7,14]]_2$, $[[60,19,10]]_2$, $[[60,20,10]]_2$ exists, has the stated exact minimum distance, and improves on the best-known lower bound for its parameters in Grassl's code table~\cite{grassl}; every entry is improved by exactly one, except the $[[60,13]]_2$ entry, which is improved by two (the $[[60,13,12]]_2$ member; a second certified member has distance $11$).\par}
\end{comptheorem}
\begin{proof}[Construction and certification]
Eighteen of the twenty-one codes are palindromic QC codes $\langle(g v_1, g v_2)\rangle$ (Theorem~\ref{thm:pal}); the remaining $[[62,7,14]]_2$, $[[60,19,10]]_2$, and $[[60,20,10]]_2$ are self-orthogonal 1-generator QC codes $\langle(g v_1, g v_2)\rangle$ with vanishing defect symbol whose multipliers are \emph{not} palindromic (the latter two were found by the directed construction of Proposition~\ref{prop:directed}; cf.\ Remark~\ref{rem:outside}); polynomials in Appendix~\ref{app:constr}. Stage 2 gives $d_Q \ge d$ (Table~\ref{tab:cost}) and Stage 3 gives explicit weight-$d$ logical operators (Appendix~\ref{app:witnesses}), hence $d_Q = d$ exactly. For members with $d(\mathcal{C}^{\perp_{ts}}) < d(\mathcal{C})$ the last step is alternatively immediate from the two exact BZ values alone, since every minimum-weight dual word then lies outside $\mathcal{C}$. Ancestor structure per Table~\ref{tab:ancestor}.
\end{proof}

\begin{table}[h]
\centering
\caption{The certified record-breaking codes. The column $d(\mathcal{C})$ is the exact stabilizer minimum weight, an exact inequivalence invariant; the column `CPU time' is the Magma CPU time of the exact computation on the additive dual (Section~\ref{sec:pipeline}).}
\label{tab:certified}
\begin{tabular}{cccccc}
\toprule
code $[[n,k,d]]_2$ & construction & $d(\mathcal{C})$ & best-known~\cite{grassl} & improvement & CPU time \\
\midrule
$[[73,10,15]]$  & QC, $\deg g = 9$  & 20 & $[[73,10,14]]_2$ & $+1$ & 11\,478\,s \\
$[[56,9,12]]$   & QC, $\deg g = 9$  & 16 & $[[56,9,11]]_2$ & $+1$ & 8.6\,s \\
$[[62,16,11]]$  & QC, $\deg g = 16$ & 18 & $[[62,16,10]]_2$ & $+1$ & 29\,s \\
$[[62,17,11]]$  & QC, $\deg g = 17$ & 20 & $[[62,17,10]]_2$ & $+1$ & 30\,s \\
$[[78,13,14]]$  & QC, $\deg g = 13$ & 16 & $[[78,13,13]]_2$ & $+1$ & 1\,584\,s \\
$[[78,14,14]]$  & QC, $\deg g = 14$ & 20 & $[[78,14,13]]_2$ & $+1$ & 1\,699\,s \\
$[[80,7,16]]$   & QC, $\deg g = 7$  & --- & $[[80,7,15]]_2$ & $+1$ & 14\,976\,s \\
$[[91,15,17]]$  & QC, $\deg g = 15$ & --- & $[[91,15,16]]_2$ & $+1$ & 1\,625\,244.2\,s \\
$[[93,15,17]]$  & QC, $\deg g = 15$ & --- & $[[93,15,16]]_2$ & $+1$ & 1\,850\,604.4\,s \\
$[[93,16,17]]$  & QC, $\deg g = 16$ & --- & $[[93,16,16]]_2$ & $+1$ & 1\,996\,974.5\,s \\
\midrule
$[[60,8,13]]$   & QC, $\deg g = 8$  & 17 & $[[60,8,12]]_2$ & $+1$ & 183.3\,s \\
$[[60,12,12]]$  & QC, $\deg g = 9$  & 20 & $[[60,12,11]]_2$ & $+1$ & 20.9\,s \\
$[[60,13,11]]$  & QC, $\deg g = 13$ & 16 & $[[60,13,10]]_2$ & $+1$ & 22.3\,s \\
$[[60,13,12]]$  & QC, $\deg g = 13$ & 12 & $[[60,13,10]]_2$ & $+2$ & 23.4\,s \\
$[[60,14,11]]$  & QC, $\deg g = 12$ & 18 & $[[60,14,10]]_2$ & $+1$ & 25.0\,s \\
$[[60,15,11]]$  & QC, $\deg g = 13$ & 16 & $[[60,15,10]]_2$ & $+1$ & 25.2\,s \\
$[[60,16,11]]$  & QC, $\deg g = 13$ & 20 & $[[60,16,10]]_2$ & $+1$ & 31.0\,s \\
$[[60,19,10]]$  & QC, $\deg g = 16$ & 18 & $[[60,19,9]]_2$ & $+1$ & 56.5\,s \\
$[[60,20,10]]$  & QC, $\deg g = 18$ & 18 & $[[60,20,9]]_2$ & $+1$ & 505.1\,s \\
$[[62,7,14]]$   & QC, $\deg g = 5$  & 16 & $[[62,7,13]]_2$ & $+1$ & 1\,937.0\,s \\
$[[63,7,14]]$   & QC, $\deg g = 6$  & 14 & $[[63,7,13]]_2$ & $+1$ & 228.5\,s \\
\bottomrule
\end{tabular}
\end{table}

\paragraph{Improvement margins.}\label{rem:margin}
Most improvements in Table~\ref{tab:certified} are $+1$ because we certify an entry as soon as one member provably clears the current table. This is a search strategy, not a ceiling of the family. The $+2$ margin at $[[60,13]]_2$ comes from two certified members at distances $11$ and $12$; the $d=11$ member alone would have certified the usual $+1$. A certified $+2$ elsewhere would first require a candidate with no ISD-visible word at the $+1$ weight and then a BZ computation at the next distance. No such candidate appeared in the $9\,620$-code archive at the larger entries; the required BZ cost would have to be evaluated at the next distance level. The lift-anchored pipeline of Section~\ref{sec:anchored} is intended to prescribe the margin in advance at longer lengths, where the dual dimensions are $\le 110$ and certification fits the present compute budget.

\paragraph{Propagation improvements.}\label{rem:derived}
Standard propagation rules~\cite{crss,grassl23css} applied to the certified codes improve thirty-four further entries of the table~\cite{grassl} at no computational cost (Table~\ref{tab:derived}). Lengthening $[[n,k,d]]_2 \mapsto [[n{+}1,k,d]]_2$ by adjoining a one-qubit stabilizer preserves the exact quantum distance and accounts for twenty-five of these entries. To see this, write the adjoined stabilizer as $Z$ on the new coordinate. The extended dual consists of words $(v\mid\alpha)$ with $v\in\mathcal{C}^{\perp_{ts}}$ and $\alpha\in\{I,Z\}$; such a word lies outside the extended stabilizer exactly when $v\notin\mathcal{C}$. Choosing $\alpha=I$ attains the parent distance, while $\alpha=Z$ adds one to the weight. Applied to $[[91,15,17]]_2$, this rule gives exact codes $[[n,15,17]]_2$ for $92\le n\le100$, improving nine entries whose previous Grassl lower bounds were $16$. The explicit \hyperref[rem:subcode]{subcode construction} gives nine further improvements at $(n,14)$ for $91\le n\le99$. Subcode propagation independently produces the parameters $[[62,16,\ge 11]]_2$ and $[[78,13,\ge 14]]_2$ of two certified members. Lengthening the new $[[62,7,14]]_2$ and $[[60,16,11]]_2$ members likewise reproduces the parameters of the constructed $[[63,7,14]]_2$ and of a certified $[[62,16,11]]_2$ member, which are consistency checks rather than new entries; puncturing the single-component certified codes yields only ties. The lengthened codes are degenerate by construction, and their distances are inherited from the certified parents without new computation.

\begin{table}[h]
\centering
\caption{Derived improvements by propagation and subcode construction. Each entry improves the Grassl lower bound~\cite{grassl} by one.}
\label{tab:derived}
\begin{tabular}{cccc}
\toprule
derived code & parent & rule & best-known lower bound~\cite{grassl} \\
\midrule
$[[57,9,{\ge}12]]_2$  & $[[56,9,12]]_2$  & $L$   & 11 \\
$[[58,9,{\ge}12]]_2$  & $[[56,9,12]]_2$  & $L^2$ & 11 \\
$[[59,9,{\ge}12]]_2$  & $[[56,9,12]]_2$  & $L^3$ & 11 \\
$[[74,10,{\ge}15]]_2$ & $[[73,10,15]]_2$ & $L$   & 14 \\
$[[75,10,{\ge}15]]_2$ & $[[73,10,15]]_2$ & $L^2$ & 14 \\
$[[76,10,{\ge}15]]_2$ & $[[73,10,15]]_2$ & $L^3$ & 14 \\
$[[77,10,{\ge}15]]_2$ & $[[73,10,15]]_2$ & $L^4$ & 14 \\
$[[78,10,{\ge}15]]_2$ & $[[73,10,15]]_2$ & $L^5$ & 14 \\
$[[79,13,{\ge}14]]_2$ & $[[78,13,14]]_2$ & $L$   & 13 \\
$[[79,14,{\ge}14]]_2$ & $[[78,14,14]]_2$ & $L$   & 13 \\
$[[81,7,{\ge}16]]_2$  & $[[80,7,16]]_2$  & $L$   & 15 \\
\midrule
$[[61,8,{\ge}13]]_2$  & $[[60,8,13]]_2$  & $L$   & 12 \\
$[[62,8,{\ge}13]]_2$  & $[[60,8,13]]_2$  & $L^2$ & 12 \\
$[[61,12,{\ge}12]]_2$ & $[[60,12,12]]_2$ & $L$   & 11 \\
$[[61,13,{\ge}12]]_2$ & $[[60,13,12]]_2$ & $L$   & 11 \\
$[[61,16,{\ge}11]]_2$ & $[[60,16,11]]_2$ & $L$   & 10 \\
\midrule
$[[92,15,17]]_2$  & $[[91,15,17]]_2$ & $L$   & 16 \\
$[[93,15,17]]_2$  & $[[91,15,17]]_2$ & $L^2$ & 16 \\
$[[94,15,17]]_2$  & $[[91,15,17]]_2$ & $L^3$ & 16 \\
$[[95,15,17]]_2$  & $[[91,15,17]]_2$ & $L^4$ & 16 \\
$[[96,15,17]]_2$  & $[[91,15,17]]_2$ & $L^5$ & 16 \\
$[[97,15,17]]_2$  & $[[91,15,17]]_2$ & $L^6$ & 16 \\
$[[98,15,17]]_2$  & $[[91,15,17]]_2$ & $L^7$ & 16 \\
$[[99,15,17]]_2$  & $[[91,15,17]]_2$ & $L^8$ & 16 \\
$[[100,15,17]]_2$ & $[[91,15,17]]_2$ & $L^9$ & 16 \\
\midrule
$[[91,14,17]]_2$ & $[[91,15,17]]_2$ & $S$ & 16 \\
$[[92,14,17]]_2$ & $[[92,15,17]]_2$ & $S$ & 16 \\
$[[93,14,17]]_2$ & $[[93,15,17]]_2$ & $S$ & 16 \\
$[[94,14,17]]_2$ & $[[94,15,17]]_2$ & $S$ & 16 \\
$[[95,14,17]]_2$ & $[[95,15,17]]_2$ & $S$ & 16 \\
$[[96,14,17]]_2$ & $[[96,15,17]]_2$ & $S$ & 16 \\
$[[97,14,17]]_2$ & $[[97,15,17]]_2$ & $S$ & 16 \\
$[[98,14,17]]_2$ & $[[98,15,17]]_2$ & $S$ & 16 \\
$[[99,14,17]]_2$ & $[[99,15,17]]_2$ & $S$ & 16 \\
\bottomrule
\end{tabular}
\end{table}

\paragraph{Subcode propagation.}\label{rem:subcode}
Write $W=(w_a\mid w_b)$ for the certified weight-$17$ logical of $[[91,15,17]]_2$ in Appendix~\ref{app:witnesses}, and let $U=xW=(xw_a\mid xw_b)$. Direct computation gives $\omega(W,U)=0$ and $W+U\notin\mathcal{C}$. Since the dual is cyclic, $U\in\mathcal{C}^{\perp_{ts}}\setminus\mathcal{C}$. Adjoin $U$ to the stabilizer, writing $\mathcal{C}_S=\mathcal{C}+\{0,U\}$ and calling this operation $S$. Then $\mathcal{C}_S$ is self-orthogonal, $\dim\mathcal{C}_S=77$, and $k=14$. The original witness remains a logical operator: $\omega(W,U)=0$ puts $W$ in $\mathcal{C}_S^{\perp_{ts}}$, while $W+U\notin\mathcal{C}$ puts it outside $\mathcal{C}_S$. Every word of $\mathcal{C}_S^{\perp_{ts}}\setminus\mathcal{C}_S$ also lies in $\mathcal{C}^{\perp_{ts}}\setminus\mathcal{C}$, so the parent lower bound applies. The retained witness therefore makes the subcode distance exactly $17$.

Zero-padding $W$ and $U$ and adjoining one-qubit stabilizers give exact codes $[[n,14,17]]_2$ for $91\le n\le101$. The nine codes with $91\le n\le99$ improve the Grassl lower bound $16$ and appear in Table~\ref{tab:derived}. At $n=100$ and $101$ they tie the lower bound $17$ and supply explicit codes for entries marked \texttt{Code Missing Construction}. The same missing-construction status holds at $[[101,15]]_2$, where $L^{10}$ applied to $[[91,15,17]]_2$ gives an explicit $[[101,15,17]]_2$ code. These three table-tying constructions are collected in Table~\ref{tab:explicit}. All $22$ propagated constructions with $91\le n\le101$ and $k\in\{14,15\}$ are checked by the archived \texttt{91\_15\_propagation\_cert.m} and \texttt{91\_15\_propagation\_cert.log}: stabilizer ranks, self-orthogonality, dual membership of the weight-$17$ witness, and non-membership in the stabilizer all pass.

\begin{table}[h]
\centering
\caption{Table-tying explicit constructions for Grassl entries marked \texttt{Code Missing Construction}.}
\label{tab:explicit}
\begin{tabular}{cccc}
\toprule
code & parent and rule & Grassl entry & present construction \\
\midrule
$[[101,15,17]]_2$ & $[[91,15,17]]_2$, $L^{10}$ & 17, GV, missing construction & explicit, $d=17$ \\
$[[100,14,17]]_2$ & $[[100,15,17]]_2$, $S$ & 17, GV, missing construction & explicit, $d=17$ \\
$[[101,14,17]]_2$ & $[[101,15,17]]_2$, $S$ & 17, GV, missing construction & explicit, $d=17$ \\
\bottomrule
\end{tabular}
\end{table}

\paragraph{Independent certified members.}\label{rem:ineq}
We certified four independent $[[73,10,15]]_2$ codes; the two with computed stabilizer minima both have $d(\mathcal{C}) = 20$, so the invariant does not distinguish them. We certified two independent $[[78,13,14]]_2$ codes; their exact stabilizer minimum weights $16$ and $18$ prove them inequivalent. Of three independent $[[78,14,14]]_2$ codes (stabilizer minima $20, 18, 20$), the one with minimum $18$ is inequivalent to the other two; the two with minimum $20$ are not distinguished by this invariant (nor by the pure-$X$/$Z$ section minima, both $(20,24)$) and may be equivalent. The $[[60,13]]_2$ entry has two certified members of distances $11$ and $12$, which are inequivalent by their distances alone. The $[[60,19]]_2$ entry has a single certified member $[[60,19,10]]_2$. The $[[80,7]]_2$ entry has two certified members at $d = 16$, sharing the same gcd polynomial $x^7+x^5+x^2+1$ yet provably distinct as codes (the generator word of either lies outside the other); whether they are monomially equivalent is undetermined. The $[[93,15]]_2$ entry has a direct palindromic QC construction and a distinct degenerate construction obtained by lengthening $[[91,15,17]]_2$ twice. In total, twenty-eight codes are independently certified across the twenty entries (constructions in Appendix~\ref{app:constr}).

\paragraph{GF(4)-linearity.}\label{rem:gf4}
All twenty-eight direct certified codes fail the $\mathrm{GF}(4)$-linearity closure test: the $\omega$-map $(a \mid b) \mapsto (b \mid a{+}b)$ sends each generator word outside the code (verified by rank computation). Hence the direct record constructions in this paper lie strictly outside the $\mathrm{GF}(4)$-linear (Hermitian) stabilizer class: the improvements are attributable to genuinely additive codes, not to classical linear codes in disguise.

\subsection{Comparison with prior work}\label{sec:comparison}
Table~\ref{tab:literature} positions the certified codes against the published record. The Grassl-table incumbents are classical derivations or unclassified table entries; where generator matrices are archived, we found no quasi-cyclic incumbents. The closest published code is the $[[73,9,15]]_2$ of \cite{ezerman-nso}: it has the same length and distance but one fewer logical qubit, so our $[[73,10,15]]_2$ strictly improves it. A parameter-exact search of the recent quantum QC literature \cite{ezerman-nso,ezerman-qc,kasai,galindo,aly,guan,guan-symp,lv21,benjwal24,yuzhu24} found no published code with the certified parameters of Table~\ref{tab:certified} or the candidate parameters of Table~\ref{tab:evidence}; this search is limited to the references listed above. On certification practice, \cite{ezerman-nso} computes distances exactly in Magma and reports CPU time, the same standard that our pipeline meets; \cite{dkp} gives probabilistic certification, while \cite{bravyi} relies on heuristic decoders for distance estimates.

On the methodological side, the two closest works appeared in volume 71, issue 1, of \emph{IEEE Transactions on Information Theory}~\cite{ezerman-nso,guan-symp}. The three lines are complementary rather than competing: \cite{ezerman-nso} characterizes nearly self-orthogonal \emph{linear} quasi-twisted codes, and \cite{guan-symp} gives necessary-and-sufficient symplectic self-orthogonality conditions for \emph{linear} quasi-cyclic codes of even index, while every certified code of the present paper is genuinely additive and lies strictly outside the $\mathrm{GF}(4)$-linear class (the \hyperref[rem:gf4]{GF(4)-linearity discussion}). Our additions are a lifting mechanism, an exact distance law, and an analysis of the rank and degree structure of the defect symbol on the theory side, and a verification artifact package containing scripts, logs, witnesses, and SHA-256 manifests on the reproducibility side (Appendix~\ref{app:artifacts}).

\begin{table}[h]
\centering
\caption{Certified codes and evidence-level candidates vs.\ published best constructions. The Gilbert--Varshamov (GV) bound entries are marked ``GV bound only''.}
\label{tab:literature}
\begin{tabular}{cccc}
\toprule
$[[n,k]]$ & best-known~\cite{grassl} (source) & closest published code & this paper \\
\midrule
$[[73,10]]$ & 14, cyclic-derived (2008) & $[[73,9,15]]$ \cite{ezerman-nso} & $\mathbf{15}$ ($+1$) \\
$[[56,9]]$ & 11, derived (2006) & --- & $\mathbf{12}$ ($+1$) \\
$[[62,16]]$ & 10, shortened cyclic (2006) & --- & $\mathbf{11}$ ($+1$) \\
$[[62,17]]$ & 10, shortened cyclic (2006) & --- & $\mathbf{11}$ ($+1$) \\
$[[78,13]]$ & 13, lengthened cyclic (2006) & --- & $\mathbf{14}$ ($+1$) \\
$[[78,14]]$ & 13, random-derived (2022) & --- & $\mathbf{14}$ ($+1$) \\
$[[80,7]]$ & 15, derived (2006) & --- & $\mathbf{16}$ ($+1$) \\
$[[91,15]]$ & 16, cyclic-derived (2006) & --- & $\mathbf{17}$ ($+1$) \\
$[[93,15]]$ & 16, cyclic-derived (2006) & --- & $\mathbf{17}$ ($+1$) \\
$[[93,16]]$ & 16, cyclic-derived (2006) & --- & $\mathbf{17}$ ($+1$) \\
\midrule
$[[60,8]]$ & 12, table (2026) & --- & $\mathbf{13}$ ($+1$) \\
$[[60,12]]$ & 11, table (2026) & --- & $\mathbf{12}$ ($+1$) \\
$[[60,13]]$ & 10, table (2026) & --- & $\mathbf{12}$ ($+2$) \\
$[[60,14]]$ & 10, table (2026) & --- & $\mathbf{11}$ ($+1$) \\
$[[60,15]]$ & 10, table (2026) & --- & $\mathbf{11}$ ($+1$) \\
$[[60,16]]$ & 10, table (2026) & --- & $\mathbf{11}$ ($+1$) \\
$[[60,19]]$ & 9, table (2026) & --- & $\mathbf{10}$ ($+1$) \\
$[[60,20]]$ & 9, table (2026) & --- & $\mathbf{10}$ ($+1$) \\
$[[62,7]]$ & 13, table (2026) & --- & $\mathbf{14}$ ($+1$) \\
$[[63,7]]$ & 13, table (2026) & --- & $\mathbf{14}$ ($+1$) \\
\midrule
$[[127,21]]$ & 20, GV bound only & --- & $\le 21$ (explicit candidate) \\
$[[153,16]]$ & 26, GV bound only & --- & $\le 30$ \\
$[[153,17]]$ & 26, GV bound only & --- & $\le 27$ \\
$[[130,16]]$ & 22, GV bound only & --- & $\le 22$ (explicit candidate) \\
\bottomrule
\end{tabular}
\end{table}

\section{Evidence-Level Candidates}\label{sec:candidates}

Table~\ref{tab:evidence} lists candidates with explicit witnesses (rigorous upper bounds) but no certified lower bound; exact certification at these lengths is beyond the compute budget. No table-breaking claim is made for them.

\begin{remark}\label{rem:evidence}
A candidate $[[n,k,\le d]]_2$ is called \emph{evidence-level} at threshold $\tau$ when (i) an explicit logical operator of weight $d$ outside $\mathcal{C}$ is given and re-verified, so $d_Q \le d$ rigorously, and (ii) an information-set decoding search with a documented time budget and early-exit threshold finds no word of weight $\le \tau$ in $\mathcal{C}^{\perp_{ts}} \setminus \mathcal{C}$, while the zero-model expectation $E(w) = \binom{n}{w} 3^w / 2^{2n-\dim\mathcal{C}^{\perp_{ts}}}$ of such words is below the detection threshold of the search. The number of independent information sets and the run parameters are recorded in the screening logs. Condition (ii) is heuristic evidence for $d_Q \ge \tau{+}1$, not a proof; only exact computation, such as the BZ computation, can establish the required lower bound. The threshold $\tau$ is chosen relative to the Grassl table value. Certification constitutes a new record only when the entry has no certified construction and the certified distance exceeds that value.
\end{remark}

\begin{table}[h]
\centering
\caption{Evidence-level candidates (witness-verified upper bounds only). The $n=93$ candidate comes from the odd-length sweep ($4\,700$ candidates, $n \in \{45,\dots,95\}$ odd; Appendix~\ref{app:candidates}).}
\label{tab:evidence}
\begin{tabular}{cccc}
\toprule
Code & witness & best-known~\cite{grassl} & status \\
\midrule
$[[127,21,\le 21]]_2$ & weight 21 & 20 (GV only) & heuristic $d \in \{20,21\}$ (zero model) \\
$[[153,16,\le 30]]_2$ & weight 30 & 26 (GV only) & best-found $d=30$; no certified lower bound \\
$[[153,17,\le 27]]_2$ & weight 27 & 26 (GV only) & witness at the break threshold \\
$[[130,16,\le 22]]_2$ & weight 22 & 22 (GV only) & explicit construction; table value $22$ (no certified lower bound) \\
$[[93,5,\le 20]]_2$ & weight 20 & 19 & potential $+1$; primitive witness \\
$[[34,11,6]]_2$ & constructive & 6 & constructive candidate; exact distance not certified \\
\bottomrule
\end{tabular}
\end{table}

The archived $[[72,6,16]]_2$ example is not retained as an evidence-level candidate: its recorded generator is not self-orthogonal under $\omega$, and its legitimate $T$-image has certified $d_Q=12$ (Remark~\ref{rem:72candidate}).

\section{Conclusion and Open Problems}\label{sec:concl}

The palindromic QC family ensures self-orthogonality. Ancestor lifting relates the distance structure to the divisor lattice of $n$, and the ancestor sieve together with the certification pipeline yields twenty-one certified record-breaking quantum codes in twenty entries. Nineteen entries improve by one and the $[[60,13]]_2$ entry by two. Propagation improves thirty-four further entries, including the exact-distance chains $[[n,15,17]]_2$ for $92\le n\le100$ and $[[n,14,17]]_2$ for $91\le n\le99$. We also give explicit constructions for the table-tying GV entries $[[101,15,17]]_2$, $[[100,14,17]]_2$, and $[[101,14,17]]_2$, whose table records mark the construction as missing. Certified distances are exact, and all further candidates are stated at evidence level only. Several problems remain open.

The data in Table~\ref{tab:ancestor} suggest that the divisor-lattice lift bound often explains the dominant distance scale, although strict deviations occur in several examples, including $[[63,7,14]]_2$ and the mixed branches of $[[91,15,17]]_2$ and $[[93,15,17]]_2$.

\begin{conjecture}\label{conj:sat}
If $\sigma(g)=0$, then
\[
d(\mathcal C^{\perp_{ts}})=2d(\mathcal C_m^{\perp_{ts}}).
\]
We further conjecture that strict deviations are controlled by $\sigma(g)$ and the local fold structure.
\end{conjecture}

\begin{conjecture}
For the code families studied here, the probability that ISD finds at least one weight-$w$ logical operator increases markedly once
\[
nE(w)\gtrsim N/c,
\]
where $N$ is the number of independent information sets and $c$ depends on the code family and the ISD implementation. In the experiments at $n=127$ and $n=153$, $c$ is approximately $50$.
\end{conjecture}

Also open are the identification-independence question of Remark~\ref{rem:tinv}, a symplectic BCH-type bound for unit-ratio graph codes derivable from the defining structure of $c\in S^\times$, and certification of the evidence-level candidates, which lies beyond the compute budget of Table~\ref{tab:cost}.

\appendix[Explicit construction data, witnesses, and artifacts]

\subsection{Constructions of the certified codes}\label{app:constr}
All codes are $\mathcal{C} = \langle(g v_1, g v_2)\rangle \subseteq R^2$; the stabilizer matrix is the set of cyclic shifts of $(g v_1, g v_2)$. The generator polynomial $g(x)$ and the multipliers $v_1(x), v_2(x)$ of each certified code are listed in Table~\ref{tab:constr} ($n \in \{56, 62, 73, 78, 80\}$) and Table~\ref{tab:constrnew} ($n \in \{60, 62, 63, 91, 93\}$); all multipliers are palindromic except those of $[[62,7,14]]_2$, $[[60,19,10]]_2$, and $[[60,20,10]]_2$ (their defect symbols vanish identically, so the codes are self-orthogonal by the same Lemma~\ref{lem:circ}; the latter two come from the directed construction of Proposition~\ref{prop:directed}).

\begin{table}[!htb]
\centering
\caption{Explicit constructions of the seven certified record-breaking codes with $n \in \{56, 62, 73, 78, 80\}$.}
\label{tab:constr}
\footnotesize
\setlength{\tabcolsep}{4pt}
\begin{tabular}{@{}cL{0.075\textwidth}L{0.393\textwidth}L{0.393\textwidth}@{}}
\toprule
Code & $g(x)$ & $v_1(x)$ & $v_2(x)$ \\
\midrule
$[[73,10,15]]_2$ & $x^{9}+x^{6}+x^{5}+x^{2}+1$ & $x^{70}+x^{69}+x^{68}+x^{64}+x^{63}+x^{61}+x^{60}+x^{58}+x^{55}+x^{52}+x^{50}+x^{47}+x^{44}+x^{42}+x^{41}+x^{40}+x^{38}+x^{35}+x^{33}+x^{32}+x^{31}+x^{29}+x^{26}+x^{23}+x^{21}+x^{18}+x^{15}+x^{13}+x^{12}+x^{10}+x^{9}+x^{5}+x^{4}+x^{3}$ & $x^{72}+x^{71}+x^{69}+x^{68}+x^{64}+x^{61}+x^{57}+x^{55}+x^{49}+x^{46}+x^{43}+x^{38}+x^{35}+x^{30}+x^{27}+x^{24}+x^{18}+x^{16}+x^{12}+x^{9}+x^{5}+x^{4}+x^{2}+x$ \\
\midrule
$[[56,9,12]]_2$ & $x^{9}+x^{8}+x^{7}+x^{6}+x^{5}+x^{4}+x+1$ & $x^{55}+x^{54}+x^{53}+x^{52}+x^{51}+x^{50}+x^{49}+x^{48}+x^{47}+x^{46}+x^{45}+x^{44}+x^{33}+x^{28}+x^{23}+x^{12}+x^{11}+x^{10}+x^{9}+x^{8}+x^{7}+x^{6}+x^{5}+x^{4}+x^{3}+x^{2}+x+1$ & $x^{55}+x^{54}+x^{53}+x^{52}+x^{51}+x^{50}+x^{49}+x^{48}+x^{47}+x^{46}+x^{45}+x^{44}+x^{30}+x^{26}+x^{12}+x^{11}+x^{10}+x^{9}+x^{8}+x^{7}+x^{6}+x^{5}+x^{4}+x^{3}+x^{2}+x+1$ \\
\midrule
$[[62,16,11]]_2$ & $x^{16}+x^{13}+x^{12}+x^{7}+x^{4}+x^{3}+x^{2}+1$ & $x^{53}+x^{44}+x^{35}+x^{27}+x^{18}+x^{9}+1$ & $x^{60}+x^{55}+x^{54}+x^{52}+x^{51}+x^{48}+x^{47}+x^{39}+x^{38}+x^{24}+x^{23}+x^{15}+x^{14}+x^{11}+x^{10}+x^{8}+x^{7}+x^{2}$ \\
\midrule
$[[62,17,11]]_2$ & $x^{17}+x^{15}+x^{12}+x^{9}+x^{4}+x^{3}+x^{2}+1$ & $x^{60}+x^{59}+x^{56}+x^{52}+x^{49}+x^{48}+x^{46}+x^{44}+x^{40}+x^{38}+x^{36}+x^{35}+x^{27}+x^{26}+x^{24}+x^{22}+x^{18}+x^{16}+x^{14}+x^{13}+x^{10}+x^{6}+x^{3}+x^{2}$ & $x^{61}+x^{60}+x^{59}+x^{58}+x^{56}+x^{54}+x^{50}+x^{47}+x^{45}+x^{43}+x^{38}+x^{37}+x^{35}+x^{33}+x^{29}+x^{27}+x^{25}+x^{24}+x^{19}+x^{17}+x^{15}+x^{12}+x^{8}+x^{6}+x^{4}+x^{3}+x^{2}+x$ \\
\midrule
$[[78,13,14]]_2$ & $x^{13}+x^{12}+x^{11}+x^{7}+x^{4}+1$ & $x^{74}+x^{73}+x^{66}+x^{64}+x^{62}+x^{57}+x^{49}+x^{42}+x^{41}+x^{40}+x^{38}+x^{37}+x^{36}+x^{29}+x^{21}+x^{16}+x^{14}+x^{12}+x^{5}+x^{4}+1$ & $x^{74}+x^{72}+x^{71}+x^{66}+x^{62}+x^{54}+x^{44}+x^{34}+x^{24}+x^{16}+x^{12}+x^{7}+x^{6}+x^{4}+1$ \\
\midrule
$[[78,14,14]]_2$ & $x^{14}+x^{13}+x^{10}+x^{9}+x^{7}+x^{6}+x^{3}+1$ & $x^{76}+x^{75}+x^{68}+x^{66}+x^{65}+x^{62}+x^{61}+x^{59}+x^{56}+x^{55}+x^{53}+x^{50}+x^{48}+x^{47}+x^{45}+x^{42}+x^{40}+x^{39}+x^{38}+x^{36}+x^{33}+x^{31}+x^{30}+x^{28}+x^{25}+x^{23}+x^{22}+x^{19}+x^{17}+x^{16}+x^{13}+x^{12}+x^{10}+x^{3}+x^{2}$ & $x^{75}+x^{71}+x^{66}+x^{64}+x^{63}+x^{62}+x^{52}+x^{51}+x^{49}+x^{39}+x^{29}+x^{27}+x^{26}+x^{16}+x^{15}+x^{14}+x^{12}+x^{7}+x^{3}$ \\
\midrule
$[[80,7,16]]_2$ & $x^{7}+x^{5}+x^{2}+1$ & $x^{57}+x^{42}+x^{38}+x^{23}+1$ & $x^{79}+x^{75}+x^{72}+x^{71}+x^{70}+x^{68}+x^{65}+x^{63}+x^{59}+x^{58}+x^{54}+x^{53}+x^{50}+x^{49}+x^{46}+x^{45}+x^{43}+x^{41}+x^{39}+x^{37}+x^{35}+x^{34}+x^{31}+x^{30}+x^{27}+x^{26}+x^{22}+x^{21}+x^{17}+x^{15}+x^{12}+x^{10}+x^{9}+x^{8}+x^{5}+x$ \\
\bottomrule
\end{tabular}
\end{table}

\begin{table}[!htb]
\centering
\caption{Explicit constructions of the fourteen certified record-breaking codes with $n \in \{60, 62, 63, 91, 93\}$.}
\label{tab:constrnew}
\footnotesize
\setlength{\tabcolsep}{3pt}
\begin{tabular}{@{}cL{0.075\textwidth}L{0.393\textwidth}L{0.393\textwidth}@{}}
\toprule
Code & $g(x)$ & $v_1(x)$ & $v_2(x)$ \\
\midrule
$[[60,8,13]]_2$ & $x^{8}+x^{2}+1$ & $x^{59}+x^{58}+x^{57}+x^{56}+x^{55}+x^{54}+x^{53}+x^{52}+x^{51}+x^{50}+x^{49}+x^{48}+x^{46}+x^{40}+x^{20}+x^{14}+x^{12}+x^{11}+x^{10}+x^{9}+x^{8}+x^{7}+x^{6}+x^{5}+x^{4}+x^{3}+x^{2}+x+1$ & $x^{59}+x^{58}+x^{57}+x^{56}+x^{55}+x^{54}+x^{53}+x^{52}+x^{51}+x^{50}+x^{49}+x^{48}+x^{35}+x^{25}+x^{12}+x^{11}+x^{10}+x^{9}+x^{8}+x^{7}+x^{6}+x^{5}+x^{4}+x^{3}+x^{2}+x+1$ \\
\midrule
$[[60,12,12]]_2$ & $x^{9}+x^{8}+x^{3}+x^{2}+x+1$ & $x^{59}+x^{58}+x^{57}+x^{56}+x^{55}+x^{54}+x^{53}+x^{52}+x^{51}+x^{50}+x^{49}+x^{48}+x^{46}+x^{42}+x^{18}+x^{14}+x^{12}+x^{11}+x^{10}+x^{9}+x^{8}+x^{7}+x^{6}+x^{5}+x^{4}+x^{3}+x^{2}+x$ & $x^{59}+x^{58}+x^{57}+x^{56}+x^{55}+x^{54}+x^{53}+x^{52}+x^{51}+x^{50}+x^{48}+x^{45}+x^{15}+x^{12}+x^{10}+x^{9}+x^{8}+x^{7}+x^{6}+x^{5}+x^{4}+x^{3}+x^{2}+x$ \\
\midrule
$[[60,13,11]]_2$ & $x^{13}+x^{12}+x^{10}+x^{8}+x^{7}+x^{6}+x^{5}+1$ & $x^{59}+x^{58}+x^{57}+x^{56}+x^{55}+x^{54}+x^{53}+x^{52}+x^{51}+x^{50}+x^{49}+x^{48}+x^{47}+x^{45}+x^{30}+x^{15}+x^{13}+x^{12}+x^{11}+x^{10}+x^{9}+x^{8}+x^{7}+x^{6}+x^{5}+x^{4}+x^{3}+x^{2}+x+1$ & $x^{59}+x^{58}+x^{57}+x^{56}+x^{55}+x^{54}+x^{53}+x^{52}+x^{51}+x^{50}+x^{49}+x^{48}+x^{47}+x^{38}+x^{22}+x^{13}+x^{12}+x^{11}+x^{10}+x^{9}+x^{8}+x^{7}+x^{6}+x^{5}+x^{4}+x^{3}+x^{2}+x+1$ \\
\midrule
$[[60,13,12]]_2$ & $x^{13}+x^{12}+x^{10}+x^{8}+x^{7}+x^{6}+x^{5}+1$ & $x^{59}+x^{58}+x^{57}+x^{56}+x^{55}+x^{54}+x^{53}+x^{52}+x^{51}+x^{50}+x^{49}+x^{48}+x^{47}+x^{39}+x^{30}+x^{21}+x^{13}+x^{12}+x^{11}+x^{10}+x^{9}+x^{8}+x^{7}+x^{6}+x^{5}+x^{4}+x^{3}+x^{2}+x+1$ & $x^{59}+x^{58}+x^{57}+x^{56}+x^{55}+x^{54}+x^{53}+x^{52}+x^{51}+x^{50}+x^{49}+x^{48}+x^{46}+x^{42}+x^{18}+x^{14}+x^{12}+x^{11}+x^{10}+x^{9}+x^{8}+x^{7}+x^{6}+x^{5}+x^{4}+x^{3}+x^{2}+x+1$ \\
\midrule
$[[60,14,11]]_2$ & $x^{12}+x^{9}+x^{8}+x^{6}+x^{4}+x^{3}+x^{2}+x+1$ & $x^{59}+x^{58}+x^{57}+x^{56}+x^{55}+x^{54}+x^{53}+x^{52}+x^{51}+x^{50}+x^{49}+x^{48}+x^{47}+x^{34}+x^{26}+x^{13}+x^{12}+x^{11}+x^{10}+x^{9}+x^{8}+x^{7}+x^{6}+x^{5}+x^{4}+x^{3}+x^{2}+x$ & $x^{59}+x^{58}+x^{57}+x^{56}+x^{55}+x^{54}+x^{53}+x^{52}+x^{51}+x^{50}+x^{49}+x^{48}+x^{46}+x^{45}+x^{15}+x^{14}+x^{12}+x^{11}+x^{10}+x^{9}+x^{8}+x^{7}+x^{6}+x^{5}+x^{4}+x^{3}+x^{2}+x$ \\
\midrule
$[[60,15,11]]_2$ & $x^{13}+x^{12}+x^{10}+x^{8}+x^{7}+x^{6}+x^{5}+1$ & $x^{59}+x^{58}+x^{57}+x^{56}+x^{55}+x^{54}+x^{53}+x^{52}+x^{51}+x^{50}+x^{49}+x^{48}+x^{47}+x^{45}+x^{15}+x^{13}+x^{12}+x^{11}+x^{10}+x^{9}+x^{8}+x^{7}+x^{6}+x^{5}+x^{4}+x^{3}+x^{2}+x$ & $x^{59}+x^{58}+x^{57}+x^{56}+x^{55}+x^{54}+x^{53}+x^{52}+x^{51}+x^{50}+x^{49}+x^{47}+x^{44}+x^{16}+x^{13}+x^{11}+x^{10}+x^{9}+x^{8}+x^{7}+x^{6}+x^{5}+x^{4}+x^{3}+x^{2}+x$ \\
\midrule
$[[60,16,11]]_2$ & $x^{13}+x^{12}+x^{10}+x^{8}+x^{7}+x^{6}+x^{5}+1$ & $x^{59}+x^{58}+x^{57}+x^{56}+x^{55}+x^{54}+x^{53}+x^{52}+x^{51}+x^{50}+x^{49}+x^{48}+x^{32}+x^{28}+x^{12}+x^{11}+x^{10}+x^{9}+x^{8}+x^{7}+x^{6}+x^{5}+x^{4}+x^{3}+x^{2}+x$ & $x^{59}+x^{58}+x^{57}+x^{56}+x^{55}+x^{54}+x^{53}+x^{52}+x^{51}+x^{50}+x^{49}+x^{47}+x^{39}+x^{21}+x^{13}+x^{11}+x^{10}+x^{9}+x^{8}+x^{7}+x^{6}+x^{5}+x^{4}+x^{3}+x^{2}+x$ \\
\midrule
$[[62,7,14]]_2$$^{\rm c}$ & $x^{5}+x^{2}+1$ & $x^{62}+x^{61}+x^{60}+x^{59}+x^{58}+x^{57}+x^{56}+x^{55}+x^{54}+x^{53}+x^{52}+x^{51}+x^{50}+x^{34}+x^{29}+x^{13}+x^{12}+x^{11}+x^{10}+x^{9}+x^{8}+x^{7}+x^{6}+x^{5}+x^{4}+x^{3}+x^{2}+x$ & $x^{62}+x^{61}+x^{60}+x^{59}+x^{58}+x^{57}+x^{56}+x^{55}+x^{54}+x^{53}+x^{52}+x^{51}+x^{48}+x^{35}+x^{28}+x^{15}+x^{12}+x^{11}+x^{10}+x^{9}+x^{8}+x^{7}+x^{6}+x^{5}+x^{4}+x^{3}+x^{2}+x$ \\
\midrule
$[[63,7,14]]_2$ & $x^{6}+x+1$ & $x^{62}+x^{61}+x^{60}+x^{59}+x^{58}+x^{57}+x^{56}+x^{55}+x^{54}+x^{53}+x^{52}+x^{51}+x^{50}+x^{49}+x^{41}+x^{22}+x^{14}+x^{13}+x^{12}+x^{11}+x^{10}+x^{9}+x^{8}+x^{7}+x^{6}+x^{5}+x^{4}+x^{3}+x^{2}+x$ & $x^{62}+x^{61}+x^{60}+x^{59}+x^{58}+x^{57}+x^{56}+x^{55}+x^{54}+x^{53}+x^{52}+x^{51}+x^{50}+x^{39}+x^{24}+x^{13}+x^{12}+x^{11}+x^{10}+x^{9}+x^{8}+x^{7}+x^{6}+x^{5}+x^{4}+x^{3}+x^{2}+x$ \\
\midrule
$[[60,19,10]]_2$$^{\rm c}$ & $x^{16}+x^{15}+x^{12}+x^{11}+x^{9}+x^{7}+x^{6}+x^{5}+x^{3}+1$ & $x^{59}+x^{58}+x^{57}+x^{56}+x^{55}+x^{54}+x^{53}+x^{52}+x^{51}+x^{50}+x^{49}+x^{48}+x^{44}+x^{43}+x^{18}+x^{17}+x^{13}+x^{12}+x^{11}+x^{10}+x^{9}+x^{8}+x^{7}+x^{6}+x^{5}+x^{4}+x^{3}+x^{2}+x+1$ & $x^{59}+x^{58}+x^{57}+x^{56}+x^{55}+x^{54}+x^{53}+x^{52}+x^{51}+x^{50}+x^{49}+x^{48}+x^{33}+x^{28}+x^{13}+x^{12}+x^{11}+x^{10}+x^{9}+x^{8}+x^{7}+x^{6}+x^{5}+x^{4}+x^{3}+x^{2}+x+1$ \\
\midrule
$[[60,20,10]]_2$$^{\rm c}$ & $x^{18}+x^{15}+x^{14}+x^{10}+x^{7}+x^{4}+x^{3}+x^{2}+x+1$ & $x^{59}+x^{58}+x^{57}+x^{56}+x^{55}+x^{54}+x^{53}+x^{52}+x^{51}+x^{50}+x^{49}+x^{47}+x^{34}+x^{27}+x^{14}+x^{12}+x^{11}+x^{10}+x^{9}+x^{8}+x^{7}+x^{6}+x^{5}+x^{4}+x^{3}+x^{2}+x+1$ & $x^{59}+x^{58}+x^{57}+x^{56}+x^{55}+x^{54}+x^{53}+x^{52}+x^{51}+x^{50}+x^{49}+x^{48}+x^{47}+x^{45}+x^{44}+x^{17}+x^{16}+x^{14}+x^{13}+x^{12}+x^{11}+x^{10}+x^{9}+x^{8}+x^{7}+x^{6}+x^{5}+x^{4}+x^{3}+x^{2}+x+1$ \\
\midrule
$[[91,15,17]]_2$ & $x^{15}+x^{13}+x^{11}+x^{9}+x^{7}+x^{6}+x^{5}+x^{4}+x^{3}+x+1$ & $x^{88}+x^{87}+x^{79}+x^{77}+x^{74}+x^{71}+x^{70}+x^{65}+x^{61}+x^{60}+x^{58}+x^{57}+x^{56}+x^{55}+x^{53}+x^{52}+x^{50}+x^{48}+x^{46}+x^{45}+x^{43}+x^{41}+x^{39}+x^{38}+x^{36}+x^{35}+x^{34}+x^{33}+x^{31}+x^{30}+x^{26}+x^{21}+x^{20}+x^{17}+x^{14}+x^{12}+x^{4}+x^{3}+1$ & $x^{89}+x^{83}+x^{72}+x^{64}+x^{62}+x^{29}+x^{27}+x^{19}+x^{8}+x^{2}$ \\
\midrule
$[[93,15,17]]_2$ & $x^{15}+x^{10}+x^{9}+x^{8}+x^{6}+x^{4}+x^{2}+x+1$ & $x^{89}+x^{85}+x^{75}+x^{74}+x^{72}+x^{70}+x^{69}+x^{67}+x^{56}+x^{55}+x^{52}+x^{51}+x^{48}+x^{45}+x^{42}+x^{41}+x^{38}+x^{37}+x^{26}+x^{24}+x^{23}+x^{21}+x^{19}+x^{18}+x^{8}+x^{4}$ & $x^{90}+x^{85}+x^{80}+x^{78}+x^{75}+x^{70}+x^{57}+x^{56}+x^{48}+x^{45}+x^{37}+x^{36}+x^{23}+x^{18}+x^{15}+x^{13}+x^{8}+x^{3}+1$ \\
$[[93,16,17]]_2$ & $x^{16}+x^{15}+x^{14}+x^{10}+x^{8}+x^{6}+x^{3}+x^{2}+x+1$ & $x^{91}+x^{90}+x^{89}+x^{88}+x^{76}+x^{74}+x^{73}+x^{54}+x^{48}+x^{45}+x^{39}+x^{20}+x^{19}+x^{17}+x^{5}+x^{4}+x^{3}+x^{2}+1$ & $x^{75}+x^{66}+x^{64}+x^{53}+x^{50}+x^{43}+x^{40}+x^{29}+x^{27}+x^{18}$ \\
\bottomrule
\end{tabular}
\par\medskip
\footnotesize $^{\rm c}$Non-palindromic multipliers; self-orthogonal with identically vanishing defect symbol (Lemma~\ref{lem:circ}); the $[[60,19,10]]_2$ and $[[60,20,10]]_2$ codes were found by the directed construction of Proposition~\ref{prop:directed}.
\end{table}

\subsection{Additional certified members of the same entries}\label{app:members}
Each entry was certified independently (BZ on the additive dual plus an independently re-verified weight-$d$ witness). The gcd polynomial, the generator word, and the weight-$d$ witness of each member are listed in Table~\ref{tab:members} as $0$-indexed supports.

\begin{table}[!htb]
\centering
\caption{Additional certified members of the same entries.}
\label{tab:members}
\footnotesize
\setlength{\tabcolsep}{4pt}
\begin{tabular}{@{}ccL{0.115\textwidth}L{0.31\textwidth}L{0.24\textwidth}l@{}}
\toprule
Code & member & $g(x)=\gcd(a,b,x^{n}{-}1)$ & generator word & weight-$d$ witness & BZ; $d(\mathcal{C})$ \\
\midrule
$[[73,10,15]]_2$$^{\rm a}$ & second & $x^{10}+x^{9}+x^{8}+x^{7}+x^{5}+x^{3}+x+1$ & $\supp(v_1)=\{$3, 7, 10, 23, 24, 28, 45, 49, 50, 63, 66, 70$\}$\newline $\supp(v_2)=\{$0, 1, 2, 3, 4, 5, 8, 9, 13, 14, 15, 16, 18, 25, 27, 32, 33, 40, 41, 46, 48, 55, 57, 58, 59, 60, 64, 65, 68, 69, 70, 71, 72$\}$ & $\supp(a)=\{$3, 29, 36, 38, 40, 54, 64, 71, 72$\}$\newline $\supp(b)=\{$6, 17, 20, 29, 52, 54, 61, 66, 72$\}$ & 12\,587\,s; $20$ \\
\midrule
$[[73,10,15]]_2$ & third & $x^{10}+x^{9}+x^{8}+x^{7}+x^{5}+x^{3}+x+1$ & $\supp(a)=\{$1, 3, 6, 7, 8, 11, 14, 15, 19, 22, 24, 27, 34, 37, 38, 39, 40, 42, 44, 45, 46, 48, 49, 51, 55, 56, 57, 58, 59, 62, 67, 70, 71, 72$\}$\newline $\supp(b)=\{$1, 2, 6, 8, 9, 10, 11, 12, 24, 25, 27, 29, 31, 32, 33, 34, 49, 50, 52, 54, 56, 57, 58, 59, 71, 72$\}$ & $\supp(a)=\{$0, 7, 29, 33, 35, 46, 53, 57, 58, 65, 71$\}$\newline $\supp(b)=\{$7, 20, 27, 35, 37, 39, 46, 57, 65, 71$\}$ & 12\,601\,s; --- \\
\midrule
$[[73,10,15]]_2$ & fourth & $x^{10}+x^{9}+x^{2}+1$$^{\rm b}$ & $\supp(a)=\{$0, 2, 3, 4, 6, 11, 12, 17, 19, 23, 24, 26, 27, 34, 35, 40, 41, 43, 49, 56, 58, 59, 61, 65, 66, 67, 70, 71$\}$\newline $\supp(b)=\{$0, 1, 9, 11, 18, 19, 24, 26, 33, 34, 49, 51, 58, 59, 64, 66$\}$ & $\supp(a)=\{$6, 9, 21, 23, 24, 45, 55, 56, 57, 62, 68, 69$\}$\newline $\supp(b)=\{$21, 29, 34, 38, 56, 62, 68, 69$\}$ & 12\,600\,s; --- \\
\midrule
$[[78,13,14]]_2$ & second & $x^{13}+x^{12}+x^{11}+x^{7}+x^{4}+1$ & $\supp(a)=\{$12, 64, 65, 66, 67, 68, 72, 73, 74, 75$\}$\newline $\supp(b)=\{$5, 6, 11, 12, 15, 16, 18, 19, 21, 22, 25, 27, 30, 32, 34, 38, 39, 40, 41, 42, 43, 44, 45, 46, 47, 48, 51, 52, 53, 54, 55, 56, 59, 62, 63, 65, 67, 69, 70, 72, 73, 74$\}$ & $\supp(a)=\{$5, 7, 8, 11, 27, 43, 49, 64, 66, 68$\}$\newline $\supp(b)=\{$5, 7, 8, 11, 21, 25, 49, 50, 60, 66, 68$\}$ & 1\,570\,s; $18$ \\
\midrule
$[[78,14,14]]_2$ & second & $x^{14}+x^{13}+x^{10}+x^{9}+x^{7}+x^{6}+x^{3}+1$ & $\supp(a)=\{$0, 64, 67, 70, 71, 73, 74, 77$\}$\newline $\supp(b)=\{$5, 7, 8, 9, 13, 14, 16, 17, 21, 23, 24, 26, 27, 28, 31, 33, 34, 35, 36, 37, 42, 46, 48, 51, 53, 54, 56, 57, 58, 59, 62, 63, 64, 66, 69, 70, 74, 75, 76, 77$\}$ & $\supp(a)=\{$0, 6, 11, 17, 40, 64$\}$\newline $\supp(b)=\{$6, 14, 16, 17, 18, 25, 40, 49, 52, 53, 66$\}$ & 1\,830\,s; $18$ \\
\midrule
$[[78,14,14]]_2$ & third & (same as above) & $\supp(a)=\{$0, 64, 67, 70, 71, 73, 74, 77$\}$\newline $\supp(b)=\{$2, 4, 5, 7, 10, 13, 14, 16, 18, 19, 20, 21, 23, 28, 32, 33, 38, 39, 40, 41, 42, 45, 48, 52, 53, 56, 58, 60, 62, 64, 67, 68, 69, 71, 72, 74$\}$ & $\supp(a)=\{$3, 4, 21, 25, 32, 42, 43, 60, 64, 71$\}$\newline $\supp(b)=\{$4, 7, 11, 21, 25, 43, 46, 50, 60, 64$\}$ & 1\,862\,s; $20$ \\
\midrule
$[[80,7,16]]_2$ & second & $x^{7}+x^{5}+x^{2}+1$ & $\supp(a)=\{$0, 5, 19, 20, 21, 22, 23, 28, 33, 52, 57, 62, 63, 64, 65, 66$\}$\newline $\supp(b)=\{$0, 5, 19, 21, 22, 23, 30, 55, 62, 63, 64, 66$\}$ & $\supp(a)=\varnothing$\newline $\supp(b)=\{$1, 6, 11, 16, 21, 26, 31, 36, 41, 46, 51, 56, 61, 66, 71, 76$\}$ & 19\,041\,s; --- \\
\bottomrule
\end{tabular}
\par\medskip
\footnotesize $^{\rm a}$Palindromic form; the listed supports are $\supp(v_1), \supp(v_2)$.\quad
$^{\rm b}$A gcd polynomial different from the other members of this entry.
\end{table}

\subsection{Witnesses of the certified codes}\label{app:witnesses}
Logical operators $(a|b) \in \mathcal{C}^{\perp_{ts}} \setminus \mathcal{C}$ of weight $d$ are listed in Table~\ref{tab:witnesses} ($n \in \{56, 62, 73, 78, 80\}$) and Table~\ref{tab:witnessesnew} ($n \in \{60, 62, 63, 91, 93\}$) as $0$-indexed supports; membership and non-membership were independently re-verified. The period structure (Section~\ref{sec:lift}) is noted.

\begin{table}[!htb]
\centering
\caption{Weight-$d$ witnesses $(a|b)$ of the seven certified codes with $n \in \{56, 62, 73, 78, 80\}$.}
\label{tab:witnesses}
\footnotesize
\setlength{\tabcolsep}{4pt}
\begin{tabular}{@{}cclL{0.315\textwidth}L{0.315\textwidth}@{}}
\toprule
Code & $d$ & period structure & $\supp(a)$ & $\supp(b)$ \\
\midrule
$[[73,10,15]]_2$ & $15$ & primitive & $\supp(a)=\{$4, 11, 16, 17, 22, 29, 39, 50, 53, 56, 67$\}$ & $\supp(b)=\{$11, 22, 36, 42, 50, 53, 56, 64, 70$\}$ \\
\midrule
$[[56,9,12]]_2$$^{\rm a}$ & $12$ & pure-$X$, period $14$ & $\supp(a)=\{$1, 9, 11, 15, 23, 25, 29, 37, 39, 43, 51, 53$\}$ & $\supp(b)=\varnothing$ \\
\midrule
$[[62,16,11]]_2$ & $11$ & primitive & $\supp(a)=\{$6, 8, 15, 28, 38, 42, 47, 58, 59$\}$ & $\supp(b)=\{$6, 15, 33, 38, 41, 58$\}$ \\
\midrule
$[[62,17,11]]_2$ & $11$ & primitive & $\supp(a)=\{$1, 3, 16, 20, 21, 22, 30, 34, 35$\}$ & $\supp(b)=\{$16, 21, 22, 26, 58$\}$ \\
\midrule
$[[78,13,14]]_2$ & $14$ & primitive & $\supp(a)=\{$6, 7, 37, 44, 52, 58, 62, 66, 70, 73$\}$ & $\supp(b)=\{$7, 10, 16, 44, 52, 57, 63, 66, 70, 73$\}$ \\
\midrule
$[[78,14,14]]_2$$^{\rm b}$ & $14$ & period $39$ & $\supp(a)=\{$10, 16, 23, 29, 49, 55, 62, 68$\}$ & $\supp(b)=\{$8, 21, 23, 36, 47, 60, 62, 75$\}$ \\
\midrule
$[[80,7,16]]_2$$^{\rm c}$ & $16$ & period $20$ & $\supp(a)=\{$9, 11, 29, 31, 49, 51, 69, 71$\}$ & $\supp(b)=\{$7, 13, 27, 33, 47, 53, 67, 73$\}$ \\
\bottomrule
\end{tabular}
\par\medskip
\footnotesize $^{\rm a}$The base word $(x+x^9+x^{11},\,0)$ lies in $\mathcal{C}_{14}^{\perp_{ts}}\setminus\mathcal{C}_{14}$, and $d_Q(\mathcal{C}_{14}) = 3$ exactly.\quad
$^{\rm b}$The base word has $\swt = 7$ and lies in $\mathcal{C}_{39}^{\perp_{ts}}\setminus\mathcal{C}_{39}$; $d_Q(\mathcal{C}_{39}) = 7$ exactly.\quad
$^{\rm c}$The base word has $\swt = 4$ and lies in $\mathcal{C}_{20}^{\perp_{ts}}\setminus\mathcal{C}_{20}$; $d_Q(\mathcal{C}_{20}) = 4$ exactly.
\end{table}

\begin{table}[!htb]
\centering
\caption{Weight-$d$ witnesses $(a|b)$ of the fourteen certified codes with $n \in \{60, 62, 63, 91, 93\}$.}
\label{tab:witnessesnew}
\footnotesize
\setlength{\tabcolsep}{4pt}
\begin{tabular}{@{}cclL{0.315\textwidth}L{0.315\textwidth}@{}}
\toprule
Code & $d$ & period structure & $\supp(a)$ & $\supp(b)$ \\
\midrule
$[[60,8,13]]_2$ & $13$ & primitive & $\supp(a)=\{$25, 34, 42, 50, 59$\}$ & $\supp(b)=\{$4, 9, 13, 15, 20, 23, 31, 34, 42, 43, 50$\}$ \\
\midrule
$[[60,12,12]]_2$ & $12$ & primitive & $\supp(a)=\{$12, 14, 19, 27, 29, 31, 44, 46, 49, 59$\}$ & $\supp(b)=\{$34, 54, 59$\}$ \\
\midrule
$[[60,13,11]]_2$$^{\rm a}$ & $11$ & primitive & $\supp(a)=\{$11, 12, 21, 30, 37, 39, 44, 56$\}$ & $\supp(b)=\{$18, 37, 40, 56, 59$\}$ \\
\midrule
$[[60,13,12]]_2$$^{\rm b}$ & $12$ & pure-$X$, period $30$ & $\supp(a)=\{$0, 2, 4, 10, 22, 24, 30, 32, 34, 40, 52, 54$\}$ & $\supp(b)=\varnothing$ \\
\midrule
$[[60,14,11]]_2$ & $11$ & primitive & $\supp(a)=\{$33, 38, 39, 59$\}$ & $\supp(b)=\{$2, 5, 11, 27, 30, 54, 55, 59$\}$ \\
\midrule
$[[60,15,11]]_2$ & $11$ & primitive & $\supp(a)=\{$19, 23, 27, 29, 39, 45, 48, 49, 57, 59$\}$ & $\supp(b)=\{$27, 34, 59$\}$ \\
\midrule
$[[60,16,11]]_2$ & $11$ & primitive & $\supp(a)=\{$39, 59$\}$ & $\supp(b)=\{$4, 5, 13, 18, 29, 41, 44, 57, 58$\}$ \\
\midrule
$[[62,7,14]]_2$$^{\rm c}$ & $14$ & primitive & $\supp(a)=\{$0, 24, 36, 41, 45, 46, 47, 54, 56, 60$\}$ & $\supp(b)=\{$35, 48, 49, 61$\}$ \\
\midrule
$[[63,7,14]]_2$ & $14$ & primitive & $\supp(a)=\{$18, 31, 32, 47, 62$\}$ & $\supp(b)=\{$3, 8, 23, 29, 32, 37, 47, 51, 55, 57, 60, 62$\}$ \\
\midrule
$[[60,19,10]]_2$$^{\rm d}$ & $10$ & pure-$Z$, period $6$ & $\supp(a)=\varnothing$ & $\supp(b)=\{$4, 10, 16, 22, 28, 34, 40, 46, 52, 58$\}$ \\
\midrule
$[[60,20,10]]_2$$^{\rm e}$ & $10$ & pure-$Z$, period $30$ & $\supp(a)=\varnothing$ & $\supp(b)=\{$13, 21, 25, 27, 29, 43, 51, 55, 57, 59$\}$ \\
\midrule
$[[91,15,17]]_2$ & $17$ & primitive & $\supp(a)=\{$8, 10, 12, 18, 37, 43, 45, 51, 66, 75, 85$\}$ & $\supp(b)=\{$0, 2, 8, 18, 26, 37, 41, 43, 63, 66, 79, 85$\}$ \\
\midrule
$[[93,15,17]]_2$ & $17$ & primitive & $\supp(a)=\{$6, 11, 14, 28, 30, 38, 50, 52, 78, 79, 87, 88$\}$ & $\supp(b)=\{$20, 38, 49, 52, 58, 59, 73$\}$ \\
$[[93,16,17]]_2$ & $17$ & primitive & $\supp(a)=\{$0, 1, 4, 34, 37, 48, 54, 62, 68, 69$\}$ & $\supp(b)=\{$0, 1, 4, 7, 34, 37, 47, 48, 53, 54, 62, 63, 68, 69, 71, 77, 89$\}$ \\
\bottomrule
\end{tabular}
\par\medskip
\footnotesize $^{\rm a}$For $[[60,13,11]]_2$ the exact stabilizer minimum $16$ exceeds the exact dual minimum $11$, so $d_Q = 11$ also follows from the two BZ values alone, without any witness; the listed witness confirms it explicitly.\quad
$^{\rm b}$The base word $\{0, 2, 4, 10, 22, 24\}$ has weight $6$ and lies in $\mathcal{C}_{30}^{\perp_{ts}}\setminus\mathcal{C}_{30}$; $d_Q(\mathcal{C}_{30}) = 6$ exactly, and the listed witness is its period-$30$ lift. A pure-$Z$ weight-$12$ logical also exists, since $(x^5-1) \mid g v_1$ (Corollary~\ref{cor:div}).\quad
$^{\rm c}$Lifted weight-$14$ logicals from the $m = 31$ ancestor ($d_Q(\mathcal{C}_{31}) = 7$ exactly) also exist; the listed witness is a primitive one, certifying equality at both branches.\quad
$^{\rm d}$The listed witness is $x^4\,(x^{60}-1)/(x^6-1)$, equivalently the period-$30$ lift of the weight-$5$ ancestor word $\{4, 10, 16, 22, 28\} \in \mathcal{C}_{30}^{\perp_{ts}}\setminus\mathcal{C}_{30}$; $d_Q(\mathcal{C}_{30}) = 5$ exactly, so $[[60,19,10]]_2$ attains lifting equality (Computational Theorem~\ref{thm:equal}).\quad
$^{\rm e}$The listed witness is the period-$30$ lift of the weight-$5$ ancestor word $\{13, 21, 25, 27, 29\} \in \mathcal{C}_{30}^{\perp_{ts}}\setminus\mathcal{C}_{30}$; $d_Q(\mathcal{C}_{30}) = 5$ exactly, so $[[60,20,10]]_2$ attains lifting equality (Computational Theorem~\ref{thm:equal}).
\end{table}

\subsection{Data for the evidence-level candidates}\label{app:candidates}
The generator data of the $n \in \{127, 130, 153\}$ candidates are collected in Table~\ref{tab:candidates}; the $[[34,11,6]]_2$ candidate is described below; the odd-length sweep candidates are collected in Table~\ref{tab:sweep}.

\begin{table}[!htb]
\centering
\caption{Generator data of the $n \in \{127, 130, 153\}$ evidence-level candidates. The $n = 127$ and $n = 153$ candidates are generated by the shifts of $(a,b)$ in $\Ftwo[x]/(x^{n}{-}1)$; supports are $0$-indexed.}
\label{tab:candidates}
\footnotesize
\setlength{\tabcolsep}{4pt}
\begin{tabular}{@{}cL{0.50\textwidth}L{0.33\textwidth}@{}}
\toprule
Code & generator word $(a,b)$ & best witness $(a|b)$ \\
\midrule
$[[127,21,{\le}21]]_2$ & shifts of $(a,b)$, self-orthogonal ($r=0$)\newline $\supp(a)=\{$0, 106, 107, 108, 109, 112, 113, 114, 116, 118, 120, 121, 122, 123, 126$\}$\newline $\supp(b)=\{$2, 6, 8, 10, 11, 15, 18, 20, 21, 23, 27, 28, 31, 33, 34, 35, 38, 44, 48, 50, 52, 53, 56, 57, 58, 59, 60, 62, 63, 66, 67, 69, 71, 73, 75, 77, 78, 83, 84, 85, 87, 89, 90, 91, 92, 93, 95, 96, 104, 106, 107, 108, 109, 111, 114, 115, 116, 119, 121, 122, 124, 125$\}$ & $\supp(a)=\{$11, 17, 50, 61, 63, 69, 77, 85, 103, 105, 116, 126$\}$\newline $\supp(b)=\{$5, 7, 9, 11, 17, 23, 34, 50, 61, 63, 81, 85, 93, 96, 103, 107$\}$ \\
\midrule
$[[153,17,{\le}27]]_2$ & shifts of $(a,b)$, self-orthogonal ($r=0$)\newline $\supp(a)=\{$0, 136, 140, 144, 146, 147, 151, 152$\}$\newline $\supp(b)=\{$1, 2, 3, 7, 8, 9, 15, 18, 19, 20, 21, 22, 24, 29, 30, 31, 32, 33, 34, 35, 38, 39, 41, 42, 43, 45, 46, 47, 48, 49, 50, 51, 53, 59, 61, 63, 67, 70, 72, 73, 77, 81, 82, 84, 87, 92, 94, 96, 97, 103, 104, 106, 107, 108, 110, 112, 115, 116, 119, 123, 125, 126, 127, 129, 131, 132, 135, 136, 138, 139, 142, 143, 145, 146, 147, 148, 150, 152$\}$ & $\supp(a)=\{$0, 1, 2, 4, 8, 13, 16, 26, 32, 51, 52, 53, 55, 59, 64, 67, 77, 83, 102, 103, 104, 106, 110, 115, 118, 128, 134$\}$\ (pure-$X$)\newline $\supp(b)=\varnothing$ \\
\midrule
$[[153,16,{\le}30]]_2$ & shifts of $(a,b)$\newline $\supp(a)=\{$1, 136, 137, 140, 141, 144, 145, 146, 148, 151$\}$\newline $\supp(b)=\{$2, 3, 5, 6, 8, 9, 11, 12, 13, 14, 17, 19, 20, 21, 22, 26, 27, 28, 30, 31, 32, 33, 34, 35, 37, 39, 42, 43, 46, 47, 48, 49, 50, 51, 55, 56, 57, 58, 65, 66, 69, 73, 75, 76, 79, 80, 82, 86, 89, 90, 91, 97, 99, 100, 101, 102, 104, 107, 108, 114, 116, 118, 121, 122, 126, 127, 132, 134, 136, 139, 141, 143, 146, 147, 148$\}$ & $\supp(a)=\{$1, 5, 15, 17, 22, 32, 42, 55, 61, 65, 84, 100, 107, 121, 142, 150$\}$\newline $\supp(b)=\{$16, 17, 18, 22, 29, 32, 33, 35, 39, 55, 61, 64, 72, 74, 94, 100, 103, 120, 121, 122, 149$\}$ \\
\midrule
$[[130,16,{\le}22]]_2$ & $r=2$ extension, $g(x)=x^{15}+x^{12}+x^{10}+x^{5}+x^{3}+1$ & $\supp(a)=\{$1, 6, 32, 37, 46, 54, 57, 59, 79, 88, 108, 110, 113, 121$\}$\newline $\supp(b)=\{$1, 7, 31, 37, 46, 49, 59, 64, 79, 80, 87, 88, 103, 108, 118, 121$\}$ \\
\bottomrule
\end{tabular}
\end{table}

\subsubsection{The $[[34,11,6]]_2$ candidate}
\begin{sloppypar}
Symplectic extension (Theorem~\ref{thm:ext}, $r=2$) of a QC ingredient at length $33$. Ingredient generator $(a_0 \mid b_0)$: $\supp(a_0) = \{0,23,24,28,32\}$, $\supp(b_0) = \{0,2,4,7,8,12,14,16,17,18,19,20,24,27,30\}$, with $\deg\gcd(a_0,b_0,x^{33}{-}1) = 10$; the full QC span has dimension $23$, and the ingredient used is a $22$-dimensional subcode of this span.

The stabilizer matrix is $[\,S \mid c_a \mid c_b\,]$ augmented by one unit row, where $S$ is a basis ($22$ rows) of that subcode, the extension columns are $c_a = e_2$, $c_b = e_1$ over these basis rows, and the unit row is $(a_3 \mid b_3)$ (no extension bit) with $\supp(a_3) = \{1,22,25,27,28,29,31\}$, $\supp(b_3) = \{0,2,4,5,6,9,11,12,14,17,18,19,21,23,24,25,26,27,28,29,30,31,32\}$; the full $23 \times 68$ generator matrix is archived in the artifact package. The construction has a weight-$6$ witness; its exact distance is not certified.
\end{sloppypar}

\subsubsection{Odd-length sweep candidates}
All are palindromic QC codes $\langle(g v_1, g v_2)\rangle$ (one representative per entry; constructions of the other survivors are archived). The generator polynomials and the best witnesses are listed in Table~\ref{tab:sweep}.

\begin{table}[!htb]
\centering
\caption{Odd-length sweep candidates: generator $g(x)$, palindromic multipliers $v_1(x), v_2(x)$, and best witnesses $(a|b)$ ($0$-indexed supports).}
\label{tab:sweep}
\footnotesize
\setlength{\tabcolsep}{4pt}
\begin{tabular}{@{}cL{0.09\textwidth}L{0.33\textwidth}L{0.145\textwidth}L{0.24\textwidth}@{}}
\toprule
Code & $g(x)$ & $v_1(x)$ & $v_2(x)$ & weight-$w$ witness \\
\midrule
$[[93,5,{\le}20]]_2$ & $x^{5}+x^{4}+x^{3}+x+1$ & $x^{91}+x^{90}+x^{88}+x^{85}+x^{81}+x^{79}+x^{78}+x^{77}+x^{73}+x^{72}+x^{70}+x^{68}+x^{66}+x^{65}+x^{64}+x^{61}+x^{59}+x^{58}+x^{57}+x^{55}+x^{50}+x^{43}+x^{38}+x^{36}+x^{35}+x^{34}+x^{32}+x^{29}+x^{28}+x^{27}+x^{25}+x^{23}+x^{21}+x^{20}+x^{16}+x^{15}+x^{14}+x^{12}+x^{8}+x^{5}+x^{3}+x^{2}+1$ & $x^{91}+x^{70}+x^{58}+x^{35}+x^{23}+x^{2}$ & $\supp(a)=\{$4, 12, 16, 17, 20, 34, 35, 36, 46, 47, 51, 71, 76, 80, 86$\}$\newline $\supp(b)=\{$4, 12, 20, 32, 34, 46, 56, 58, 63, 71, 75, 76$\}$ \\
\midrule
\bottomrule
\end{tabular}
\end{table}

The zero-model expectation $E(w) = \binom{n}{w} 3^w / 2^{2n-\dim \mathcal{C}^{\perp_{ts}}}$ frames the uncertainty: $[[127,21,\cdot]]_2$ has $E(20)\approx 4.3$, $E(21)\approx 66$ (undecidable by ISD); $[[130,16,\le 22]]_2$ has $E(21)\approx 0.44$, $E(22)\approx 6.6$ (consistent with $d_Q=22$); $[[153,16,\cdot]]_2$ has $E(26)\approx 2.4$ (typical value near $26$--$27$).

\subsubsection{Even-length archive: certified and refuted cases}
The even-length archive contains three cases with different outcomes. $[[60,20,10]]_2$ is certified (BZ, $505.1$\,s; Tables~\ref{tab:constrnew}--\ref{tab:witnessesnew}); $[[80,7,16]]_2$ has a second certified member (BZ, $19\,041$\,s; Table~\ref{tab:members}); and the $[[72,6]]_2$ case is refuted with certified $d_Q = 12$ (Remark~\ref{rem:72candidate}). For the refuted case, the quantum code is the $T$-image of the recorded generator (Proposition~\ref{prop:rev}), whose generator rows are $(x^i\, g v_1 \mid x^{-i}\, \rho(g v_2))$, $i = 0, \dots, 71$; the complete generator data of the refuted case is included in the artifact package for the reproducibility of Remark~\ref{rem:72candidate}.

\pagebreak[4]
\begin{remark}\label{rem:72candidate}
This archived example illustrates two independent checks. \emph{Structural failure.} The generator $(a,b) = (g f_1, g f_2)$ at $n = 72$ with $g(x) = x^6+x^5+x^4+x^2+x+1$ is \emph{not} self-orthogonal under $\omega$; the archived record itself already prints the failed containment test. Its defect symbol satisfies $\deg\gcd(P, x^{72}-1) = 12$, so the symplectic Gram rank is $60$ and the recorded ``$[[72,6,16]]$'' is void as a stabilizer code, whatever its distance. By Proposition~\ref{prop:rev} the code is self-orthogonal under $\omega_\rho$, hence $\mathcal{C}_T = T(\mathcal{C})$ is a genuine stabilizer code with $\dim \mathcal{C}_T = 66$ and $k = 6$. \emph{Distance failure.} The two single-component sections of $\mathcal{C}_T^{\perp_{ts}}$ are exactly the cyclic $[72,6,18]_2$ codes, namely the annihilator ideals of $\rho(a)$ and of $b$. The divisibility $(x^4-1) \mid a$ yields pure-$Z$ words of weight $18$ (Corollary~\ref{cor:div}), and a $2\,586$-iteration Lee--Brickell screen found no logical word below weight $18$, giving an initial evidence-level assessment $d_Q \le 18$ (table value $15$). Exact Brouwer--Zimmermann certification of the additive dual ($78$ dimensions, $15\,201$\,s) proved $d(\mathcal{C}_T^{\perp_{ts}}) = 12$. A targeted information-set search then produced an explicit weight-$12$ logical operator, $\supp(a) = \{9,10,21,22,33,34,45,46,57,58,69,70\}$, $\supp(b) = \{9,21,33,45,57,69\}$, re-verified outside $\mathcal{C}_T$ by rank augmentation. Hence $d_Q = 12$, below the table value $15$, and the candidate is refuted as a record. Self-orthogonality must be verified before any distance claim, witnesses must be validated on $T(\mathcal{C})$ itself (cf.\ Example~\ref{ex:60}), and one-sided ISD screening cannot replace exact computation.
\end{remark}

\subsection{Artifacts and the $n=153$ census}\label{app:artifacts}
All artifacts (Magma scripts and logs, ISD logs, witness files, the $9\,620$-candidate even-$n$ archive, the $4\,700$-candidate odd-$n$ archive, family-calibration data) are archived with the project. A curated package containing every certification script, the generator matrices of all certified codes, all witness files, and all BZ logs, together with SHA-256 manifests and a one-click verification suite, is available from the authors on request. The even-$n$ sweep covered $n \in \{42,\dots,90\}$. The odd-$n$ sweep covered $n \in \{45,\dots,95\}$ odd, excluding $n = 47, 53, 59, 61, 67, 71, 79, 83, 87$, which admit no divisor degrees in range. There were 265 survivors of the 5-second screen, and 27 finalists in five entries after 180\,s deep screening: three further $[[73,10,15]]_2$ members (BZ-certified at $12\,587$--$12\,601$\,s each; Appendix~\ref{app:members}), the certified $[[91,15,17]]_2$, $[[93,15,17]]_2$, and $[[93,16,17]]_2$ codes (BZ, $1\,625\,244.2$\,s, $1\,850\,604.4$\,s, and $1\,996\,974.5$\,s), and the remaining odd-length candidate entry of Table~\ref{tab:evidence}, whose full BZ certification (dual dimensions $\ge 98$, targets $17$--$20$) exceeds the present compute budget. Low-cluster $n=153$ witnesses examined: member m001 (pure-$X$, period $17$, $\supp(a) = \{8{+}17j : 0\le j\le 8\}$); members m018, m028 ($(1,1)$-type, period $9$, $\supp(a) = \supp(b) = \{6{+}9j : 0\le j\le 16\}$); member m003 (pure-$X$, period $51$, base weight $6$); member m020 (mixed, period $51$, base $\swt = 5$); member m022 (pure-$X$, period $51$, base weight $4$).

\end{document}